\documentclass[sigconf]{aamas}

\usepackage{balance} %

\usepackage{caption}
\usepackage{subcaption}
\usepackage{tikz}
\usetikzlibrary{graphs}
\usetikzlibrary{graphs.standard}
\usetikzlibrary{positioning,fit}
\usepackage{pgfplots}
\usepgfplotslibrary{colormaps}
\usepgfplotslibrary{external}

\usepackage{algorithm}
\usepackage{algorithmic}

\usepackage{dblfloatfix}

\usepackage{amsmath,amsfonts,bm,bbm}

\def\eqref#1{equation~(\ref{#1})}

\def\1{\bm{1}}

\usepackage{cancel}

\usepackage{listings}
\usepackage{enumerate}

\DeclareMathAlphabet{\mathsfit}{\encodingdefault}{\sfdefault}{m}{sl}
\SetMathAlphabet{\mathsfit}{bold}{\encodingdefault}{\sfdefault}{bx}{n}

\newcommand{\softmax}{\mathrm{softmax}}

\usepackage{amsthm}

\makeatletter
\newtheorem*{rep@theorem}{\rep@title}
\newcommand{\newreptheorem}[2]{%
\newenvironment{rep#1}[1]{%
 \def\rep@title{#2 \ref{##1}}%
 \begin{rep@theorem}}%
 {\end{rep@theorem}}}
\makeatother

\newtheorem{theorem}{Theorem}
\newreptheorem{theorem}{Theorem}

\newreptheorem{proposition}{Proposition}

\newreptheorem{corollary}{Corollary}
\newtheorem{lemma}{Lemma}
\newreptheorem{lemma}{Lemma}

\newtheorem{definition}{Definition}

\newcommand\myeq{\mathrel{\overset{\makebox[0pt]{\mbox{\normalfont\tiny\sffamily def}}}{=}}}

\newcommand{\bE}{\mathbb{E}}

\newcommand{\bw}{\mathbf{w}}

\newcommand{\cD}{\mathcal{D}}

\newcommand{\ie}{{\it i.e.},~}  %
\newcommand{\eg}{{\it e.g.},~}  %

\newcommand{\simplex}{\Delta}

\newtheorem*{theorem-nonumber}{Theorem}
\newtheorem*{lemma-nonumber}{Lemma}

\definecolor{darkgreen}{RGB}{0,125,0}
\definecolor{darkblue}{RGB}{0,0,125}

\newcounter{mlNoteCounter}

\newcounter{klNoteCounter}

\newcounter{gpilNoteCounter}

\newcounter{imgempNoteCounter}

\setcopyright{ifaamas}
\acmConference[AAMAS '24]{Proc.\@ of the 23rd International Conference
on Autonomous Agents and Multiagent Systems (AAMAS 2024)}{May 6 -- 10, 2024}
{Auckland, New Zealand}{N.~Alechina, V.~Dignum, M.~Dastani, J.S.~Sichman (eds.)}
\copyrightyear{2024}
\acmYear{2024}
\acmDOI{}
\acmPrice{}
\acmISBN{}

\acmSubmissionID{608}

\title{Approximating the Core via Iterative Coalition Sampling}

\author{Ian Gemp}
\affiliation{
 \institution{Google DeepMind}
 \city{London}
 \country{United Kingdom}}
\email{imgemp@google.com}

\author{Marc Lanctot}
\affiliation{
 \institution{Google DeepMind}
 \city{Montreal}
 \country{Canada}}
\email{lanctot@google.com}

\author{Luke Marris}
\affiliation{
 \institution{Google DeepMind}
 \city{London}
 \country{United Kingdom}}
\email{marris@google.com}

\author{Yiran Mao}
\affiliation{
 \institution{Google DeepMind}
 \city{London}
 \country{United Kingdom}}
\email{yiranm@google.com}

\author{Edgar Du{\'e}{\~n}ez-Guzm{\'a}n}
\affiliation{
 \institution{Google DeepMind}
 \city{London}
 \country{United Kingdom}}
\email{duenez@google.com}

\author{Sarah Perrin}
\affiliation{
 \institution{Google DeepMind}
 \city{Paris}
 \country{France}}
\email{sarahperrin@google.com}

\author{Andras Gyorgy}
\affiliation{
 \institution{Google DeepMind}
 \city{London}
 \country{United Kingdom}}
\email{agyorgy@google.com}

\author{Romuald Elie}
\affiliation{
 \institution{Google DeepMind}
 \city{Paris}
 \country{France}}
\email{relie@google.com}

\author{Georgios Piliouras}
\affiliation{
 \institution{Google DeepMind}
 \city{London}
 \country{United Kingdom}}
\email{gpil@google.com}

\author{Michael Kaisers}
\affiliation{
 \institution{Google DeepMind}
 \city{Paris}
 \country{France}}
\email{mkaisers@google.com}

\author{Daniel Hennes}
\affiliation{
 \institution{Google DeepMind}
 \city{Zürich}
 \country{Switzerland}}
\email{hennes@google.com}

\author{Kalesha Bullard}
\affiliation{
 \institution{Google DeepMind}
 \city{London}
 \country{United Kingdom}}
\email{ksbullard@google.com}

\author{Kate Larson}
\affiliation{
 \institution{Google DeepMind}
 \city{Montreal}
 \country{Canada}}
\email{katelarson@google.com}

\author{Yoram Bachrach}
\affiliation{
 \institution{Google DeepMind}
 \city{London}
 \country{United Kingdom}}
\email{yorambac@google.com}

\begin{abstract}
The core is a central solution concept in cooperative game theory, defined as the set of feasible allocations or payments such that no subset of agents has incentive to break away and form their own subgroup or coalition.  However, it has long been known that the core (and approximations, such as the least-core) are hard to compute. This limits our ability to analyze cooperative games in general, and to fully embrace cooperative game theory contributions in domains such as explainable AI (XAI), where the core can complement the Shapley values to identify influential features or instances supporting predictions by black-box models. 
We propose novel iterative algorithms for computing variants of the core, which avoid the computational bottleneck of many other approaches; namely solving large linear programs. As such, they scale better to very large problems as we demonstrate across different classes of cooperative games, including weighted voting games, induced subgraph games, and marginal contribution networks. We also explore our algorithms in the context of XAI, providing further evidence of the power of the core  for such applications.
\end{abstract}


\newcommand{\BibTeX}{\rm B\kern-.05em{\sc i\kern-.025em b}\kern-.08em\TeX}

\begin{document}

\pagestyle{fancy}
\fancyhead{}

\maketitle

\section{Introduction}

Coalitional game theory studies games played by self-interested players where incentives are aligned and fully-binding contracts are supported~\cite{Chalkiadakis12Computational}. The utility achieved depends on which group, \ie which {\it  coalition}, of players is formed to solve the game. The key problem that researchers have focused on is how this total utility should be divided among the individual members of the team. Several solution concepts have been proposed to address this question, with the core~\cite{gillies1953} and the Shapley value~\cite{Shapley51} being among the best known. The Shapley value quantifies the contribution of each individual player via the differences in value obtained over all possible subgroups which include and exclude that player.
In contrast, the core is a rough cooperative analogue of the Nash equilibrium: a utility division that is stable, \ie no subset of players has incentive to deviate from the chosen coalition. 

Cooperative game theory has found many applications, ranging from analyzing power in decision making bodies and political settings~\cite{shapley1954method}, through predicting how players may share team payoffs or costs~\cite{littlechild1977aircraft,dubey1982shapley} to risk attribution~\cite{tarashev2016risk}. One prominent application of cooperative game theory is in analyzing what drives the predictions of machine learning models. As more of our decisions are now informed by learned models, it is crucial to understand {\it why} our models output what they do, which is the main goal of explainable AI (XAI)~\cite{Barredo20Explainable}.
The Shapley value has been adopted in the machine learning community to explain model predictions~\cite{lundberg2017unified}: 
the features are treated as players, and the payoff to a coalition is the output of the model when training exclusively on the subset of input features in the coalition.
The Shapley value then computes the marginal contribution of each feature to the output, offering additive explanations for the effect of each feature on prediction.

However, there are some limitations regarding Shapley values: it may produce counter-intuitive explanations~\cite{sundararajan2020many}, behavioral studies have shown the core can be more predictive of human payment divisions, and accurate estimation requires many samples~\cite{maleki2013bounding,bachrach2010approximating,mitchell2022sampling}\footnote{See Appendix~\ref{sec:shapley_vs_core} for a more comprehensive discussion of the differences between the Shapley value and the core.}.
As such, the core has been recently motivated as a possible alternative to Shapley values for XAI~\cite{YanProcaccia21}. Unlike Shapley values, the core focuses on stability~\cite{Chalkiadakis12Computational}: a payment division is {\it stable} if no player has incentive to deviate from the coalition and form a different one on their own.
Consequently, core payments may better reflect each player's ``market value''~\cite{samet1984core,day2008core}.  

A key barrier to applying cooperative game theory in practice is the high computational cost associated with calculating the various solution concepts, which may be hard to compute even for very restricted forms of cooperative games~\cite{Chalkiadakis12Computational}. %
For some solutions, such as the Shapley value and other power indices, there exist efficient approximation algorithms~\cite{bilbao2000generating,castro2009polynomial}, e.g.,
an approximate form of the Shapley value can be computed using Monte Carlo sampling~\cite{fatima2008linear,bachrach2010approximating,mitchell2022sampling}, which might explain its wide adoption. 

In contrast, solutions based on the core have proved to be more elusive; there exist multiple approaches that can approximate or exactly solve for the core, but can only be applied to very restricted classes of coalitions games~\cite{kern2003matching,elkind2009computing,bachrach2011least,kimms2016core}. Computing the core exactly requires solving a linear program (LP) with an exponential number of constraints (one per each possible subset of the agents).
Alternatively, assuming access to incrementally sampled coalition values, bounds can be given for likely stable payments that relate to the core~\cite{balcan2015learning}. 
A recent breakthrough in approximating the core is a Monte Carlo algorithm that samples coalitions in this way, yielding probabilistic bounds on approximation quality~\cite{YanProcaccia21}. This approach significantly expands the set of games for which one can tractably approximate the core, albeit it still requires solving LPs, where the size of the LP grows linearly in the number of samples. 

{\bf Our contribution:}
In this paper, we introduce novel iterative algorithms for the core. None of our algorithms require solving a linear program: as such, they enjoy greater scaling potential to very large problems. We demonstrate this efficiency in practice across several classes of coalition games including weighted voting games~\cite{elkind2009computing}, induced subgraph games~\cite{deng_papa_1994}, marginal contribution networks~\cite{ieong2005marginal}, and apply them to feature-importance and data-valuation problems arising in XAI~\cite{Barredo20Explainable}.

\section{Background}
\label{l_preliminaries}

The central concept in cooperative game theory is that of a {\em coalition}. Given a set of $n$ agents, $I=\{1,2,\ldots, n\}$, a coalition is a subset of agents, $C\subseteq I$, where the {\em grand coalition} is the coalition containing all agents in $I$. A coalitional game, $G$,  is further coupled with a characteristic function, $v:2^I\rightarrow \mathbb{R}$ that assigns a real value to each coalition of agents, representing the total utility that the coalition of agents achieve together. We study transferable utility games, where $v(C)$ is interpreted as value that is to be shared and transferred between members of the coalition $C$. As is standard, $v(\emptyset)=0$.

The characteristic function only defines the gains a coalition can achieve; it does not define how these gains are distributed among the coalition members. An \emph{imputation} (also called payoff vector) $p = (p_1,\ldots ,p_n)$ is a division of the gains of the grand coalition $I$ among the agents, where $p_i \ge 0$, such that $\sum_{i=1}^n p_i = v(I)$. We call $p_i$ the payoff of agent $i$, and denote the payoff of a coalition $C$ as $p(C) = \sum_{i \in C} p_i$. 

A basic requirement for a good imputation is \emph{individual rationality}, which states that for all agents $i \in C$, we have $p_i \geq v(\{i\})$; otherwise, some agent has incentives to work alone instead. Similarly, we say a coalition $B$ \emph{blocks} imputation $(p_1,\ldots,p_n)$ if $p(B) < v(B)$, since the members of $B$ can split from the original coalition, derive
the gains of $v(B)$ by working together and then reallocate this \emph{excess}, $v(B)-p(B)>0$, to agents in $B$. 
This incentive to break away and form new coalitions leads to instability and has long been a focus of cooperative game theory. 
 The most prominent solution concept is that of the core~\cite{gillies1953}.

\begin{definition}
\label{l_def_core}
The core of a coalitional game $G$  is the set of all imputations that are not blocked by any coalition.  That is it contains imputation $p = (p_1,\ldots,p_n)$ if and only if for all $C\subseteq I$, $p(C) \geq v(C)$.
\end{definition}

Unfortunately, the core can be empty, meaning that for every imputation, there exists a  blocking coalition.  Thus, relaxations of the core are often studied by, for example, assuming that the gains made by forming a blocking coalition are small.  The $\epsilon$-\emph{core} allows for such  slight relaxations of the inequalities in the core definition.

\begin{definition}
\label{l_epsilon_core}
Given $\epsilon$, the $\epsilon$-core  of coalitional game $G$ is the set of all imputations such that  for any  $C \subseteq I$, $p(C) \geq v(C) - \epsilon$, i.e., $\epsilon$-core $\myeq \{ (p_1,\ldots,p_n) \,\, \vert \,\, p(C) \geq v(C) - \epsilon \,\, \forall \,\, C \subseteq I \}$.
\end{definition}
Clearly, for large enough values of $\epsilon$, the $\epsilon$-core is non-empty. Furthermore, if $\epsilon=0$, the definition of the $\epsilon$-core is equivalent to the core.
A natural question is what is the smallest $\epsilon$ such that the $\epsilon$-core is non-empty.
Given a game $G$ we consider the set $\{ \epsilon \vert$ the $\epsilon$-core of G is not empty$\}$. It is easy to see that this set is compact, and thus has a minimal element $\epsilon_{min}$. 

\begin{definition}
\label{l_least_core}
Given coalitional game $G$, define $\epsilon_{\min} = \min \{\epsilon|$ $\epsilon$-core of $G$ is non-empty$\}$. The least core of $G$ is the $\epsilon_{min}$-core of $G$, and this value $\epsilon_{min}$ is called the Least-Core Value (LCV) of $G$.
\end{definition}

\section{Approximating the Core through Iterative Payoff Adjustments via Coalition Sampling}
\label{sect_core_approx_algos}

Earlier work has shown that given a game $G$, checking if the core is non-empty is $NP$-hard~\cite{deng_papa_1994} and, furthermore, that computing the least-core exactly is impossible~\cite{Balkanski17Statistical}. This opens up the question we address in this section;  how to best approximate the least core.

The least-core can be formulated as a linear programming (LP) problem. Let $c$ be a binary coalition-membership vector for coalition $C$ such that $c_i=1$ if $i\in C$ and 0 otherwise. When clear from the context we  sometimes use $c$ to represent coalition $C$.  Furthermore, without loss of generality, assume that $v(I)=1$, re-scaling $v$ with a factor $\frac1{v(I)}$ if necessary. 
Then the least-core problem is to solve
\begin{align}
    \min_{p, \epsilon}& \quad \epsilon \label{eq:lp-obj}
    \\ s.t.& \quad v(c) - \epsilon - p^\top c \le 0 \quad \forall \,\,  c \label{eq:coalition-constraints}
    \\ & \quad\sum_{i=1}^n p_i = 1
    \\ & \quad p_i \ge 0 \quad \forall \,\, i \in I. \label{eq:nonzero-payoff}
\end{align}
\noindent The challenge is that the Constraint~\ref{eq:coalition-constraints} relates to $2^n$ coalitions, making it infeasible to solve for games with many players. A recent method to approximating the least-core randomly samples some number of the constraints imposed on the coalitions and solves the resulting LP~\cite{YanProcaccia21}. However, even when only using sub-sampled constraints, solving LPs is still time consuming.\footnote{Methods that guarantee polynomial runtime for solving LPs~\cite{karmarkar1984new} have a runtime that is cubic in the number of parameters. An efficient implementation of the Simplex method~\cite{dantzig1963linear} is much faster in practice, but is still time consuming for large LPs.}
We present alternative iterative algorithms that avoid the requirement of solving the LP, circumventing a significant computational bottleneck. 

Our algorithms are presented in increasing complexity. Our first two algorithms, Iterative Projections (Section~\ref{sec:alg:iterative_proj}) and Stochastic Subgradient Descent (Section~\ref{sec:subgrad}) take a given value of $\epsilon$ and seek an imputation in the $\epsilon$-core (assuming the $\epsilon$-core is non-empty). In order to compute the least-core, one may apply an outer loop that calls these methods so as to perform a binary search for the minimal value $\epsilon_{min}$ yielding a non-empty $\epsilon_{min}$-core. Our Core Lagrangian method (Section~\ref{sec:core-lagrangian}) directly returns the least-core, including both the $\epsilon_{min}$ (least-core value), and an imputation in the $\epsilon_{min}$-core. Our empirical results in Section~\ref{section:empirical_analysis} are reported for this algorithm alone.

\subsection{The $\epsilon$-Core via Iterative Projections}
\label{sec:alg:iterative_proj}
We observe that each Constraint~\ref{eq:coalition-constraints} on $p$ represents a half-space, namely a convex set. Assuming that there exists some $p$ such that all constraints can be satisfied for a given $\epsilon$, the corresponding constraint satisfaction problem  is to find $p$ such that Constraint~\ref{eq:coalition-constraints} is satisfied. This is equivalent to solving a convex feasibility problem by finding a $p$ on the $(n-1)$-dimensional simplex $\simplex$ that lies within the intersection of all these convex sets, a problem that can be solved via von Neumann-Halperin method of cyclic alternating projections~\citep{bauschke1996projection,deutsch1995dykstra,badea2016ritt}. We start with an initial guess for $p$ and then cycle through the constraints (for each coalition, $C$), iteratively modifying the guess by projecting it onto the feasible set represented by each individual constraint. We observe that the projection, $\texttt{Proj}_{c,\epsilon}(p)$, of an imputation $p$ onto a half-space $p^\top c \ge v(C) - \epsilon$ has a closed form solution that can be computed efficiently:
\begin{align}
  \texttt{Proj}_{c,\epsilon}(p) &= \begin{cases}
    p &   v(C) - \epsilon - p^\top c \le 0
    \\ p - (p^\top c + \epsilon - v(C)) / ||c||^2 c & \text{ otherwise}
    \end{cases} \nonumber
    \\ &= p + \frac{d_c}{\vert c \vert} c \quad \text{ where $\vert c \vert = $ the size of the coalition}\label{eqn:proj_halfspace}
\end{align}
and where $d_c = \max(0, v(C) - \epsilon - p^\top c)$ is the deficit. If $v(C) - \epsilon - p^\top c < 0$, players in coalition $C$ are being paid \emph{more} than they are actually contributing (where the payment is $\epsilon + p^\top c$). Intuitively, players are incentivized to remain in a coalition if $d_c = 0$, otherwise, the coalition is unstable. Algorithm~\ref{alg:iterative_proj} formalizes this method. 

\begin{theorem}
Let $v : 2^I \rightarrow \mathbb{R}$ be a characteristic function. Given an $\epsilon$ and assuming the $\epsilon$-core exists, Algorithm~\ref{alg:iterative_proj} converges to an $\epsilon$-core imputation asymptotically, i.e.: $\lim_{T \rightarrow \infty} p_T \rightarrow \epsilon$-core$(v)$.
\end{theorem}
\begin{proof}
We can appeal directly to classical convergence results of the von Neumann-Halperin cyclic projections algorithm~\citep{badea2016ritt}.
\end{proof}

\begin{algorithm}[t]
\caption{$\epsilon$-Core via von Neumann-Halperin}\label{alg:iterative_proj}
\begin{algorithmic}
\REQUIRE Number of iterations $T$
\STATE $p_0 = \frac{1}{n} \mathbf{1}_n$
\FOR{$t = 1 \le T$}
\STATE Select coalition $C$ (cyclically) as binary vector $c$
\STATE $d_c = \max(0, v(C) - \epsilon - p^\top c)$
\STATE $p_t \leftarrow p_{t-1} + \frac{d_c}{\vert C \vert} c$
\STATE $p_t \leftarrow \texttt{Proj}_{\simplex}(p_t)$
\ENDFOR
\ENSURE $p_T$
\end{algorithmic}
\end{algorithm}

\subsection{The $\epsilon$-Core via Stochastic (Sub)Gradient Descent}\label{sec:subgrad}

While projecting an imputation $p$ onto a hyperplane can be computed efficiently,  cycling through all constraints is not, and it is also computationally expensive to update $p$ on a batch of constraints as that involves projecting $p$ into the intersection of a set of half-spaces.
Instead, we propose a new approach that re-formulates the projection problem as an optimization problem and admits an efficient batch algorithm.

We define a loss function, $\ell_c(p,\epsilon)$ for each coalition $C$ (represented by vector $c$):
\begin{align}
    \ell_c(p, \epsilon) &= \frac{1}{2 \vert c \vert}\big(\max(0, v(C) - \epsilon - p^\top c)\big)^2 = \frac{d_c^2}{2 \vert c \vert}.
\end{align}
Observe that  if $v(C) - \epsilon - p^\top c \le 0$, $\ell_c = 0$, otherwise, we accrue some positive loss for not satisfying the constraint for coalition $c$. Furthermore, a valid negative (sub)gradient of $\ell_c$ with respect to $p$ is exactly the update direction computed by the projection $\texttt{Proj}_{c,\epsilon}$ from~\eqref{eqn:proj_halfspace}:
\begin{align}
    -\nabla_p \ell_c &= \frac{d_c}{\vert c \vert} c. \label{eqn:subgrad}
\end{align}

Therefore, if we run stochastic gradient descent (SGD) with learning rate equal to $1$ on $\sum_c \ell_c(p)$, sampling one coalition at a time,  we recover Algorithm~\ref{alg:iterative_proj}. However, it is also possible to increase the minibatch size, sampling multiple coalitions at one time, as shown in Algorithm~\ref{alg:stochastic_gradient}. Theorem~\ref{prop:sgd-approx} shows that if Algorithm~\ref{alg:stochastic_gradient} returns an imputation that is approximately in the $\epsilon$-core, then it provably lies in the $\epsilon'$-core, albeit with $\epsilon'>\epsilon$. The proof can be found in Appendix~\ref{app:theorem1proof}.

\begin{algorithm}
\caption{$\epsilon$-Core via SGD}\label{alg:stochastic_gradient}
\begin{algorithmic}
\REQUIRE Number of iterations $T$
\REQUIRE Batch size $B$
\REQUIRE Step size schedule $\eta_t$
\STATE $p = \frac{1}{n} \mathbf{1}_n$
\FOR{$t = 1 \le T$}
\STATE Sample batch $C_B$ containing $B$ coalitions
\STATE Compute average gradient over batch: $\nabla_p = \frac{1}{B} \sum_{C \in C_B} \nabla_p \ell_c$
\STATE $p \leftarrow p - \eta_t \nabla_p$
\STATE $p \leftarrow \texttt{Proj}_{\simplex}(p)$
\ENDFOR
\STATE Sample $t^* \in \{0, \ldots, T\}$ according to $P(t^* = t) = \frac{\eta_t}{\sum_{t'} \eta_{t'}}$.
\ENSURE $p_{t^*}$
\end{algorithmic}
\end{algorithm}

\begin{theorem}\label{prop:sgd-approx}
If $\ell_c(\epsilon, p) \le \gamma^2$ for all $c$, then $p$ is in the $(\epsilon+\sqrt{2n}\gamma)$-core.
\end{theorem}

\begin{theorem}
Let $v : 2^I \rightarrow \mathbb{R}$ be a characteristic function. Given an $\epsilon$ and assuming the $\epsilon$-core exists, Algorithm~\ref{alg:stochastic_gradient} converges to an $\epsilon$-core imputation in expectation at a rate of $\mathcal{O}(T^{-1/4})$.
\end{theorem}
\begin{proof}
We can appeal directly to convergence rates of stochastic projected subgradient algorithms for convex optimization~\cite{davis2018stochastic}. These rates prove Algorithm~\ref{alg:stochastic_gradient} converges to a stationary point with $\mathcal{O}(1/\sqrt{T})$ expected squared gradient norm. Recall~\eqref{eqn:subgrad} to determine that the squared gradient is proportional to $d_c^2$. Therefore, the violation of the core constraints, $d_c$, decays as $\mathcal{O}(T^{-1/4})$.
\end{proof}

\subsection{The Least-Core as a Saddle Point Problem}
\label{sec:core-lagrangian}
Algorithms~\ref{alg:iterative_proj} and~\ref{alg:stochastic_gradient} assume an $\epsilon$ is provided such that all constraints represented by Constraint~\ref{eq:coalition-constraints} are satisfied.\footnote{Recall that such an $\epsilon$ is guaranteed to exist.} However, we are really interested in finding the smallest such  $\epsilon$ for which the constraints still hold, i.e. the least-core value (LCV). To this end, we reformulate the original LP, making use of a single non-linear constraint:
\begin{align}
    \min_{p \in \simplex, \epsilon}& \quad \epsilon \label{eqn:obj}
    \\ s.t.& \quad \sum_{C \subseteq I} \ell_c(p, \epsilon) \le \gamma^2 \label{eqn:con}
\end{align}
for some constant $\gamma>0$. Note that if $\gamma = 0$, we would recover the solution to the least-core LP in~(\ref{eq:lp-obj})-(\ref{eq:nonzero-payoff}). For $\gamma > 0$, we can recover an approximate solution via Theorem~\ref{prop:sgd-approx}. While this form is no longer an LP, it retains a crucial property: convexity.  Each $\ell_c(p, \epsilon)$ is convex in $p$ and $\epsilon$, and hence the nonlinear constraint $\sum_{C \subseteq I} \ell_c(p, \epsilon) - \gamma^2$ remains convex in $p$ and $\epsilon$. Similarly  Objective~\ref{eqn:obj} is convex as it is linear in $\epsilon$ and $p$. This  allows us to view the optimization problem as a saddle-point problem using Lagrange multipliers.\footnote{Note that we introduce the constant $\gamma>0$ so as to ensure Slater's condition holds so as to meet the necessary and sufficient conditions for optimality of the subsequent saddle point solution~\cite{boyd2004convex}.}  We first observe that via Karush-Kuhn-Tucker conditions, convexity in the objective and constraint is sufficient for optimality of the solution to the corresponding Lagrangian formulation (see Sec 5.5 of~\cite{boyd2004convex}):
\begin{align}
    \min_{p \in \simplex, \epsilon \in [\underline{\epsilon}, \overline{\epsilon}]} \max_{\mu \in [0, \overline{\mu}]} \mathcal{L}(p, \epsilon, \mu) = \epsilon + \mu \Big( \sum_{C \subseteq I} \ell_c(p, \epsilon) - \gamma^2 \Big). \label{eqn:lagrangian}
\end{align}

We can bound both $\epsilon$ and $\mu$ which ensures the function values and gradients are bounded as well. Let $v_{\max} = \max_{C \subseteq I} v(C)$.
Then every coalitional constraint is trivially satisfied if $\epsilon\geq v_{\max}$, while the constraint associated with the grand coalition can not be satisfied if $\epsilon < 0$. Therefore, we can bound $\epsilon\in[\underline{\epsilon}, \overline{\epsilon}] = [0, v_{\max}]$.
For $\mu$, if only one constraint is violated, then we want $\mu$ to be large enough to force an increase in $\epsilon$. Therefore, for any single violated coalitional constraint associated with $C$, we want $\nabla_{\epsilon} \mathcal{L} = 1 - \mu (d_c / |C|) < 1 - \mu (\gamma / |C|)$ to be strictly less than zero. This implies that if $\mu$ is at least $|C| / \gamma$, then $\epsilon$ will increase in response to any violated constraint. Hence, we set $\overline{\mu} = n / \gamma$ as an upper bound.

We observe that the Lagrangian formulation is convex in the primal variables $(p, \epsilon)$ and concave in the dual variable $\mu$, hence~\eqref{eqn:lagrangian} is typically referred to as a convex-concave saddle-point problem which can be equivalently formulated as a monotone variational inequality problem~\cite{facchinei2003finite}, for which stochastic algorithms exist including extra gradient~\citep{korpelevich1976extragradient} and Stochastic Mirror Prox~\cite{juditsky2011solving}.
Algorithm~\ref{alg:lagrangian_core} provides pseudocode for the process. The key step is the Update function (Algorithm~\ref{alg:base_update}) which requires a $\texttt{Prox}_{x}$ operator and a monotone map operator $F$.\footnote{See~\citep{mokhtari2020unified} for proximal point approaches to saddle point problems.}
In variational inequality formulations $VI(F, \mathcal{X})$, the problem is to find $x^* \in \mathcal{X}$ such that $\langle F(x^*), x - x^* \rangle \ge 0$ for all $x \in \mathcal{X}$, and the map $F: \mathcal{X} \rightarrow \mathbf{R}^{n + 2}$ is typically written as a single vector valued map containing the ``descent'' directions of all variables. For our setting, we can define $F$ as follows:
\begin{align}
    F(x) &=
    \begin{bmatrix}
        \nabla_p \mathcal{L}
        \\ \nabla_{\epsilon} \mathcal{L}
        \\ -\nabla_{\mu} \mathcal{L}
    \end{bmatrix}
    =
    \begin{bmatrix}
        -\mu \Big( \sum_{C \subseteq I} \frac{d_c}{\vert C \vert} c \Big)
        \\ 1 - \mu \Big( \sum_{C \subseteq I} \frac{d_c}{\vert C \vert} \Big)
        \\ -\Big( \sum_{C \subseteq I} \ell_c(p, \epsilon) - \gamma^2 \Big) \label{eqn:vi_map}
    \end{bmatrix}
\end{align}
where $x = [p, \epsilon, \mu] \in \mathcal{X} = \simplex \times [\underline{\epsilon}, \overline{\epsilon}] \times [0, \overline{\mu}]$.
Observe that we  can use sampling of minibatches of coalitions to Monte Carlo estimate the sums in $F(x)$, e.g., $\sum_{C \subseteq I} \ell_c(p, \epsilon) = \frac{2^n}{B} \mathbb{E}_{C_B \sim I} [\sum_{C \in C_B} \ell_c(p, \epsilon)]$.

\begin{lemma}\label{lemma:monotone}
The map $F$ in~\eqref{eqn:vi_map} is monotone, i.e., $\langle F(x) - F(x'), x - x' \rangle \ge 0$ over all $x, x' \in \mathcal{X}$.
\end{lemma}
The proof of Lemma~\ref{lemma:monotone} can be found in the Appendix.

\begin{theorem}
Let $v : 2^I \rightarrow \mathbb{R}$ be a characteristic function. Algorithm~\ref{alg:lagrangian_core} reduces the duality gap $\max_{\mu'} \mathcal{L}(p, \epsilon, \mu') - \min_{p', \epsilon'} \mathcal{L}(p', \epsilon', \mu)$ at a rate of $\mathcal{O}(1/\sqrt{T})$.
\end{theorem}
\begin{proof}
We can appeal directly to convergence rates of mirror prox algorithms for monotone variational inequalities~\citep{juditsky2011solving} given we have already argued above that the norm of the map $F$ and its variance are both finite and the map $F$ is monotone (Lemma~\ref{lemma:monotone}).
\end{proof}

In experiments, we use a tailored $\texttt{Prox}$ operator, specifically,
\begin{align}
    \texttt{Prox}_{x}(\eta F(y)) &=
    \begin{bmatrix}
        \texttt{softmax}(\log(p) - \eta F_p(y)), \\
        \texttt{clip}(\epsilon - \eta F_{\epsilon}(y), \underline{\epsilon}, \overline{\epsilon}), \\
        \texttt{clip}(\mu - \eta F_{\mu}(y), 0, \overline{\mu})
    \end{bmatrix}
\end{align}
where $F_z(y)$ retrieves the part of the vector output $F(y)$ corresponding to the variable $y$ and $\texttt{softmax}(s) = [ \frac{e^{s_1}}{\sum_j e^{s_j}}, \ldots, \frac{e^{s_n}}{\sum_j e^{s_j}}]$.

Instructions for accelerating our proposed algorithms on GPUs/TPUs are in Appendix~\ref{app:accelerators}. Hyperparameters including, e.g., learning rate schedules, used in experiments are found in Appendix~\ref{app:alg_hyps}.

\begin{algorithm}[t]
\caption{\texttt{Update}}\label{alg:base_update}
\begin{algorithmic}
\REQUIRE Initial iterate $x$, map evaluation iterate $y$, batch size $B$, step size $\eta$, $\texttt{Prox}$ operator
\STATE Sample batch $C_B$ containing $B$ coalitions
\STATE Compute $F(y)$ using batch of coalitions
\STATE $x' \leftarrow \texttt{Prox}_{x}(\eta F(y))$
\ENSURE $x'$
\end{algorithmic}
\end{algorithm}

\begin{algorithm}[t]
\caption{Least-Core via Mirror Prox (Core Lagrangian)}\label{alg:lagrangian_core}
\begin{algorithmic}
\REQUIRE Number of iterations $T$, batch size $B$, step size schedule $\eta_t$
\STATE $p_0 = \frac{1}{n} \mathbf{1}_n$
\STATE $\epsilon_0 = \overline{\epsilon}$
\STATE $\mu_0 = \overline{\mu}$
\STATE $x_0 = [p_0, \epsilon_0, \mu_0]$
\FOR{$t = 1 \le T$}
\STATE $x' \leftarrow \texttt{Update}(x_{t-1}, x_{t-1}, B, \eta)$
\STATE $x_t \leftarrow \texttt{Update}(x_{t-1}, x', B, \eta)$
\ENDFOR
\STATE Compute weighted average of imputations: $p^* = \frac{\eta_t p_t}{\sum_{t'} \eta_{t'}}$.
\ENSURE $p^*$
\end{algorithmic}
\end{algorithm}

\section{Empirical Analysis}
\label{section:empirical_analysis}

We tested our algorithm, Core Lagrangian (CL) (Alg.~\ref{alg:lagrangian_core}), on a range of cooperative games and report our findings here. In our first set of experiments we compared the performance of CL to that of a recent state-of-the-art algorithm for approximating the least-core which relies on solving the LP using sampled coalitions~\cite{YanProcaccia21} (Section~\ref{sect:timing}). 
We then use our algorithm as a tool  to study the stability properties of a number of prominent cooperative-game classes, illustrating the benefits of having algorithms for approximating the core for large games (Section~\ref{sec:stability_compact_game_repr}). Finally, we examine applications of our algorithms for explainable AI (XAI) purposes (Section~\ref{sect:xai_core_empirical}).

\subsection{Timing Sampled LPs vs. Core Lagrangian}
\label{sect:timing}

\begin{figure*}[b]
\centering

    \begin{subfigure}[t]{0.19\textwidth}
        \centering
        \begin{tikzpicture}
            \begin{axis}[
                xlabel=$\sigma$,
                ylabel=$\xi$,
                ylabel near ticks,
                xlabel near ticks,
                xtick={ 0,7,14 },
                xticklabels={ 0.01,0.15,0.30 },
                ytick={ 0,7,14 },
                yticklabels={ 0.48,0.50,0.52 },
                enlargelimits=false,
                axis on top,
                colorbar,
                colorbar style={
                    xlabel near ticks,
                    yticklabel style={
                        /pgf/number format/.cd,
                        fixed,
                        precision=2,
                        fixed zerofill,
                    },
                    width=2mm,
                    at={(1.1,1)},
                },
                width=0.9\linewidth,
                height=0.9\linewidth,
            ]
                \addplot [matrix plot*,point meta=explicit] file [] {data/sweep_wvg_gauss.dat};
            \end{axis}
        \end{tikzpicture}
        \captionsetup{justification=centering}
        \caption{Gaussian ($\mu=1.0$)}\label{fig:wvg:gauss}
    \end{subfigure}\hfill%
    \begin{subfigure}[t]{0.19\textwidth}
        \centering
        \begin{tikzpicture}
            \begin{axis}[
                xlabel=$\lambda$,
                ylabel=$\xi$,
                ylabel near ticks,
                xlabel near ticks,
                xtick={ 0,7,14 },
                xticklabels={ 0.25,1.38,2.50 },
                ytick={ 0,7,14 },
                yticklabels={ 0.30,0.50,0.70 },
                enlargelimits=false,
                axis on top,
                colorbar,
                colorbar style={
                    xlabel near ticks,
                    yticklabel style={
                        /pgf/number format/.cd,
                        fixed,
                        precision=1,
                        fixed zerofill,
                    },
                    width=2mm,
                    at={(1.1,1)},
                },
                width=0.9\linewidth,
                height=0.9\linewidth,
            ]
                \addplot [matrix plot*,point meta=explicit] file [] {data/sweep_wvg_exp.dat};
            \end{axis}
        \end{tikzpicture}
        \captionsetup{justification=centering}
        \caption{Exponential}\label{fig:wvg:exp}
    \end{subfigure}\hfill%
    \begin{subfigure}[t]{0.19\textwidth}
        \centering
        \begin{tikzpicture}
            \begin{axis}[
                xlabel=$\alpha$,
                ylabel=$\beta$,
                ylabel near ticks,
                xlabel near ticks,
                xtick={ 0,7,14 },
                xticklabels={ 0.10,1.05,2.00 },
                ytick={ 0,7,14 },
                yticklabels={ 0.10,1.05,2.00 },
                enlargelimits=false,
                axis on top,
                colorbar,
                colorbar style={
                    xlabel near ticks,
                    yticklabel style={
                        /pgf/number format/.cd,
                        fixed,
                        precision=1,
                        fixed zerofill,
                    },
                    width=2mm,
                    at={(1.1,1)},
                },
                width=0.9\linewidth,
                height=0.9\linewidth,
            ]
                \addplot [matrix plot*,point meta=explicit] file [] {data/sweep_wvg_beta.dat};
            \end{axis}
        \end{tikzpicture}
        \captionsetup{justification=centering}
        \caption{Beta ($\xi = 0.25$)}\label{fig:wvg:beta}
    \end{subfigure}\hfill%
    \begin{subfigure}[t]{0.19\textwidth}
        \centering
        \begin{tikzpicture}
            \begin{axis}[
                xlabel=$\sigma$,
                ylabel=$p$,
                ylabel near ticks,
                xlabel near ticks,
                xtick={0,10,20},
                xticklabels={1.0, 2.0, 3.0},
                ytick={0,4,8,12,16,20},
                yticklabels={0.0,0.2,0.4,0.6,0.8,1.0},
                enlargelimits=false,
                axis on top,
                colorbar,
                colorbar style={
                    xlabel near ticks,
                    yticklabel style={
                        /pgf/number format/.cd,
                        fixed,
                        precision=1,
                        fixed zerofill,
                    },
                    width=2mm,
                    at={(1.1,1)},
                },
                width=0.9\linewidth,
                height=0.9\linewidth,
            ]
                \addplot [matrix plot*,point meta=explicit] file [] {data/sweep_erdos_renyi.dat};
            \end{axis}
        \end{tikzpicture}
        \captionsetup{justification=centering}
        \caption{Erdős-Rényi\\~}\label{fig:wgg:erdosrenyi}
    \end{subfigure}\hfill%
    \begin{subfigure}[t]{0.19\textwidth}
        \centering
        \begin{tikzpicture}
            \begin{axis}[
                xlabel=$p$,
                ylabel=$k$,
                ylabel near ticks,
                xlabel near ticks,
                xtick={0,10,20},
                xticklabels={0.0,0.5,1.0},
                ytick={0,4,8,12,16,20},
                yticklabels={4,8,12,16,20,24},
                enlargelimits=false,
                axis on top,
                colorbar,
                colorbar style={
                    xlabel near ticks,
                    yticklabel style={
                        /pgf/number format/.cd,
                        fixed,
                        precision=1,
                        fixed zerofill,
                    },
                    width=2mm,
                    at={(1.1,1)},
                },
                width=0.9\linewidth,
                height=0.9\linewidth,
            ]
                \addplot [matrix plot*,point meta=explicit] file [] {data/sweep_newman_watts.dat};
            \end{axis}
        \end{tikzpicture}
        \captionsetup{justification=centering}
        \caption{Newman Watts Strogatz\\($\sigma=1$)}\label{fig:wgg:nws}
    \end{subfigure}

\caption{(\subref{fig:wvg:gauss})-(\subref{fig:wvg:beta}) Heatmaps illustrating the impact of parameter axes on the least-core value (in color) in three weighted voting games. (\subref{fig:wgg:erdosrenyi})-(\subref{fig:wgg:nws}) Mean least-core value for Erdős-Rényi and Newman Watts Strogatz when sweeping over two hyperparameters. Constant hyperparameters are shown in parenthesis.}
\label{fig:games:lcv}
\end{figure*}
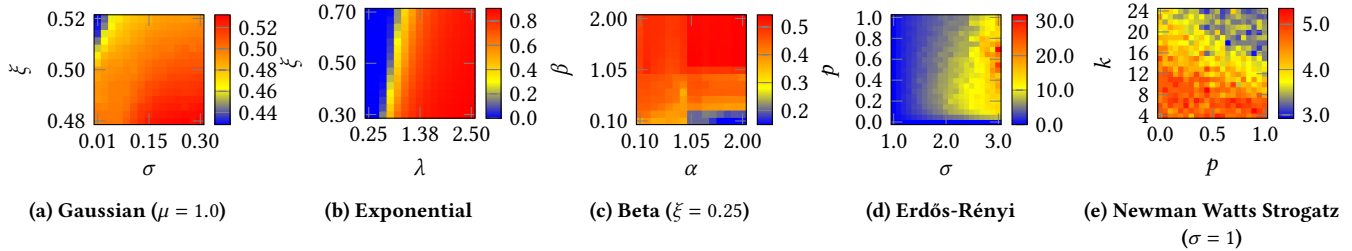

To test how efficient CL is in practice, we evaluate its performance across very large ($n=100$) instances of weighted voting games. 
\begin{definition} (\cite[Definition 4.1]{Chalkiadakis12Computational}) A weighted voting game $G$ is a tuple $(I, \bw, q)$ where $I = \{1, 2, \cdots, n\}$ is a set of agents, $\bw = (w_1, w_2, \cdots, w_n) \in \Re^n$ is a vector of weights (one per player), and $q \in \Re$ is a quota. The characteristic function is
\[ v(C) = \left\{ \begin{array}{ll}
         1 & \mbox{if $\sum_{c \in C} w_c \geq q$};\\
         0 & \mbox{otherwise}.\end{array} \right. \] 
\end{definition}

\begin{figure}[t]
\centering
\includegraphics[width=0.40\textwidth]{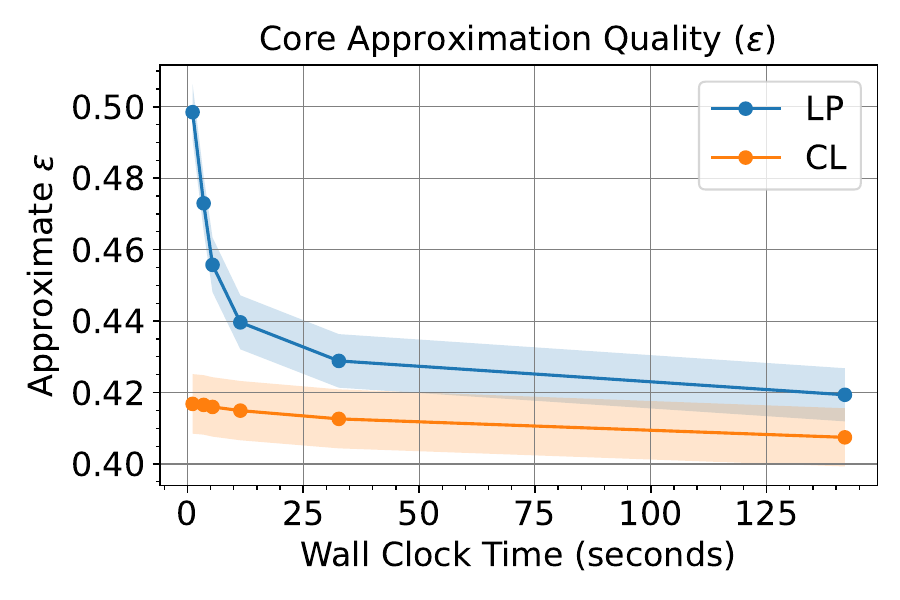}
\caption{Approximate $\epsilon$ of the linear program (LP) versus the core Lagrangian (CL) method as a function of computation time (seconds). 
The $x$-axis corresponds to wall-clock time taken by both algorithms run side-by-side given $k \in \{500, 1000, 2000, 4000, 8000, 16000\}$ sampled coalitions for the LP method, $t_k$. A better core approximation quality is reflected in having a lower $\epsilon$ for the same runtime. The $y$-axis represents the approximate $\epsilon$ of the least-core solution found, $\hat{\epsilon}(p_{LP}, \hat{C})$ and  $\hat{\epsilon}(p_{CL}, \hat{C})$, computed over the same set of $50,000$ coalitions, $\hat{C}$.
Each data point $(t_k, \hat{\epsilon}(p_x, \hat{C}))$ represents an average over the same set of $2,500$ random weighted voting games with shading indicating standard error of the mean.
\label{fig:timings-wvgs}}
\end{figure}

Weighted voting games are well studied in the literature (\eg~\cite{taylor1992characterization,elkind2009computational,zuckerman2012manipulating,aziz2011false,see2014cost,rey2010complexity}) and can model real-world scenarios such as  the European Union voting system~\cite{RaunioWiberg98,bilbao2002voting}, or multiagent resource allocation problems where an agent's weight correspond to resources available for that agent, and the threshold is the total amount of pooled resources required to accomplish some task of interest.

In our experiments, each agent $i$ has a weight $w_i \in \{1, 2, \cdots, 100\}$ drawn uniformly at random and we set the quota $q = \xi n \bE[w_i]$, where $\xi$ is the proportional threshold (fraction of the expected total weight). We use proportional thresholds $\xi$ sampled uniformly at random in the range $[0.1, 0.9]$.
There are $2^{100}$ coalitions and thus $2^{100}$ constraints in the LP formulation of the least core (Constraint~\ref{eq:coalition-constraints}), so exactly solving for the least-core is computationally infeasible.

We compared our method to Yan and Procaccia's recent algorithm for approximating the least core~\cite{YanProcaccia21}.  Their algorithm samples a subset of coalitions and returns the LP solution based on constraints built only from the sampled subset. With high probability, this   method is guaranteed to reduce the approximation error as the number of coalitions sampled increases~\cite[Theorem 1]{YanProcaccia21}. 

It is difficult to compare iterations of CL to the number of sampled coalitions used in LP; instead, we chose a number of sampled coalitions $k$ and recorded the wall clock time taken by the LP method for each value of $k$, say $t_k$ seconds. Then, we let CL run for $t_k$ seconds and retrieved the solution after that amount of elapsed time. The LP method used CVXPY~\cite{diamond2016cvxpy} while our CL code used JAX and optax~\cite{deepmind2020jax}. For each $t_k$, both methods returned an  imputation, $p_{LP}$ and $p_{CL}$, respectively. For each imputation we computed its respective (average) $\epsilon$ value. 
Due to the size of the game the error of the solution cannot be computed exactly, so we approximated the value of $\epsilon$ by sampling a set of $50,000$ coalitions, $\tilde{C}$, and computing  $\hat{\epsilon}(p, \hat{C}) = \max_{c \in \hat{C}}(v(c) - p^\top c)$. Figure~\ref{fig:timings-wvgs} summarizes our findings. In particular, we observed a significant improvement in the core approximation when using our method (CL) as opposed to the LP-based method. We did observe, however, that the threshold affects the degree to which CL outperforms the LP-based method. In particular, when the threshold was close to 50\% of the total expected weight, the LP-based method would sometimes outperform CL.

\subsection{Analyzing Stability in Compact Game Representations}
\label{sec:stability_compact_game_repr}

After validating that our algorithm can effectively solve large cooperative games,  showing improvements over a strong baseline, we further illustrate its value  by making a rigorous study across prominent classes of cooperative games to better understand how different game-features affect stability.

\subsubsection{Weighted Voting Games}
\label{sec:wvgs}
In our first set of experiments we studied how stability was impacted by different parameterizations of weighted voting games, defined in Section~\ref{sect:timing}. Since we were running multiple experiments, we set the number of agents to be $n=15$, but drew agents' weights from different distributions and varied the proportional quota, $\xi$. Figure~\ref{fig:games:lcv} presents our results. We sample 10,000 games for each parameter configuration, and examine the average least-core value in these games. 
We first note that regions that have a low least-core value (blue) are stable or very close to being stable, whereas  high $\epsilon$-core indicate regions of instability.  Second, we observe that the threshold greatly affected the stability of an instance. If the threshold was low, many coalitions form and achieve the maximum coalitional value of 1. If the threshold was very high, in expectation the only successful coalition was the grand coalition, reducing the likelihood of blocking coalitions. 

In Figure~\ref{fig:games:lcv}(a) we present results where the agents' weights were generated by a Gaussian distribution with $\mu=1.0$ and $\sigma\in[0.01,0.3]$. We observe that for a fixed quota, high weight-variance leads to less stable games. Figure~\ref{fig:games:lcv}(b) presents results when agents' weights were drawn from an exponential distribution with $\lambda\in[0.25, 2.50]$. We observe that as $\lambda$ increases we find less stable games. We hypothesize that since  agent weights are more concentrated around low values, possible successful coalitions often share agents, opening up the possibility of the formation of blocking coalitions.
Finally, in Figure~\ref{fig:games:lcv}(c) we present results where we used a Beta distribution with parameters $(\alpha,\beta)$. We notice a difference in stability of games when the parameters are either less than 1.0 or greater than 1.0. If both $\alpha$ and $\beta$ are high, games are less stable. However, there is a region of stability when $\alpha>1.0$ and $\beta<1.0$

\subsubsection{Graph Games}
\label{sec:graph-games}
In many cooperative games, relationships between agents are modelled via a graph $G=\langle V,E\rangle$ (\eg~\cite{myerson77graphs,deng_papa_1994,bachrach2007computing,demange2004group,resnick2009cost,bachrach2010path,branzei2011social,rey2011bribery}). In induced subgraph games~\cite{deng_papa_1994}, the agents are  represented by vertices of the graph and the edge weight, $w(e)$ for $e=(i,j)$, indicates the value that agents $i$ and $j$ accrue from being in the same coalition. An absence of an edge implies there is no interaction between a pair of agents and thus, no loss or benefit from being in the same coalition. Given  graph $G$, the characteristic function for the induced cooperative game is $v(C)=\sum_{e\in\{(i,j)\in E|i,j\in C\}} w(e)$.

We conduct an empirical study to better understand how the underlying graph structure influences the stability of the game. To this end, we select a collection of well-known parameterized random graph models, each with different properties, and generate 80 cooperative graph games for each parameterization, measuring their stability by computing their least-core values via our algorithm CL. We set $n=32$, and generate edge weights with a Gaussian distribution with variance $\sigma = 1.0$ (unless specified otherwise) and positive mean $\mu$ chosen such that 60\% of edge weights are positive in expectation.\footnote{If graph games have only positive edge weights, then the core is non-empty, i.e. the least-core value is 0.0~\cite{deng_papa_1994}.} All graphs were generated using the NetworkX library \citep{hagberg2008networkx}. We present results on two graph-classes, Erdős-Rényi~\cite{erdHos2013spectral} and  Newman Watts Strogatz~\cite{newman1999graph}, and discuss four additional models in Appendix~\ref{app:graph-games}.

In the Erdős-Rényi graph model there is a single parameter $p$, which indicates the probability that for any pair of vertices, $x,y$, there is an edge connecting $x$ and $y$. In our experiments we varied $p$ from 0.0 to 1.0, and varied the variance $\sigma$ of the weight distribution from 1.0 to 3.0. Our results are shown in Figure~\ref{fig:wgg:erdosrenyi}. We observe that the stability of the games generated depends on both $p$ and $\sigma$. In particular, games with high weight variance for edge weights, and high uncertainty as to whether an edge would form between any pair of vertices (\ie for $p\in[0.4,0.7])$ are less stable.

The Newman Watts Strogatz model generates graphs with the small-world property and consists of two parameters. An instance of a graph is initialized as a ring and connected with its $s\lfloor \frac{k}{2}\rfloor$ nearest neighbours. Additional edges in the graph are added with probability $p$.  Figure~\ref{fig:wgg:nws} presents our results as we varied $k$ from 4 to 24, and $p\in[0.0,1.0]$. The variance of the edge-weight distribution was fixed at $\sigma=1.0$.  In particular we observe that both parameters $p$ and $k$ positively correlate with stability in these graphs.

More generally, these results (along with further results in Appendix~\ref{app:graph-games}) show the interplay between the parameters of random-graph generators and  the stability of the induced cooperative game. We believe that further research in this area is warranted, but also emphasize that such empirical analysis can only be carried out  when using algorithms for tractably solving cooperative games at scale, such as the algorithms we have proposed.

\begin{figure*}[ht!]
    \begin{minipage}[c]{0.330\textwidth}
    \centering
    \begin{subfigure}[t]{\textwidth}
        \centering
        \caption{Boston Housing (global / local)}
        \includegraphics[width=0.48\textwidth]{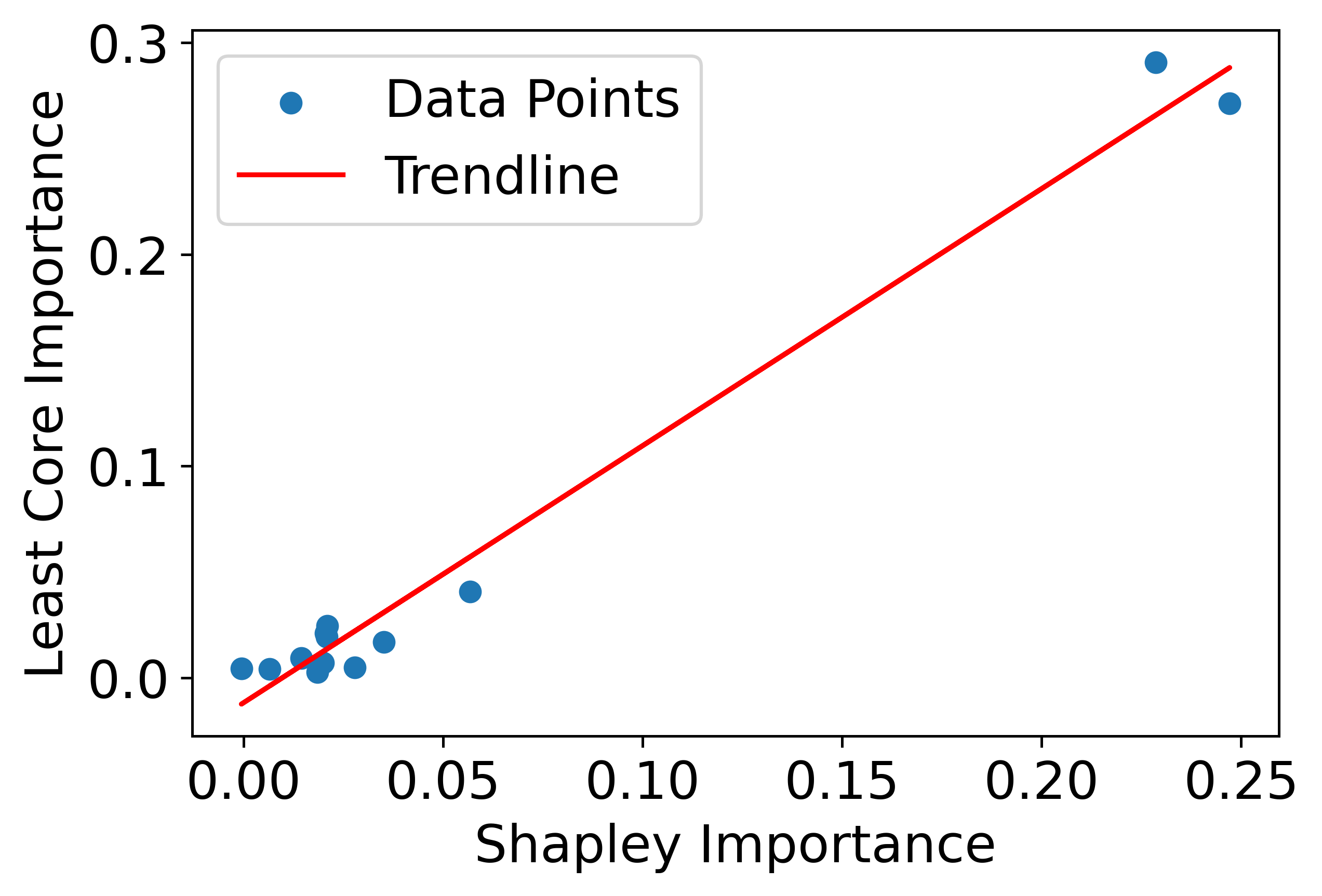}
        \includegraphics[width=0.48\textwidth]{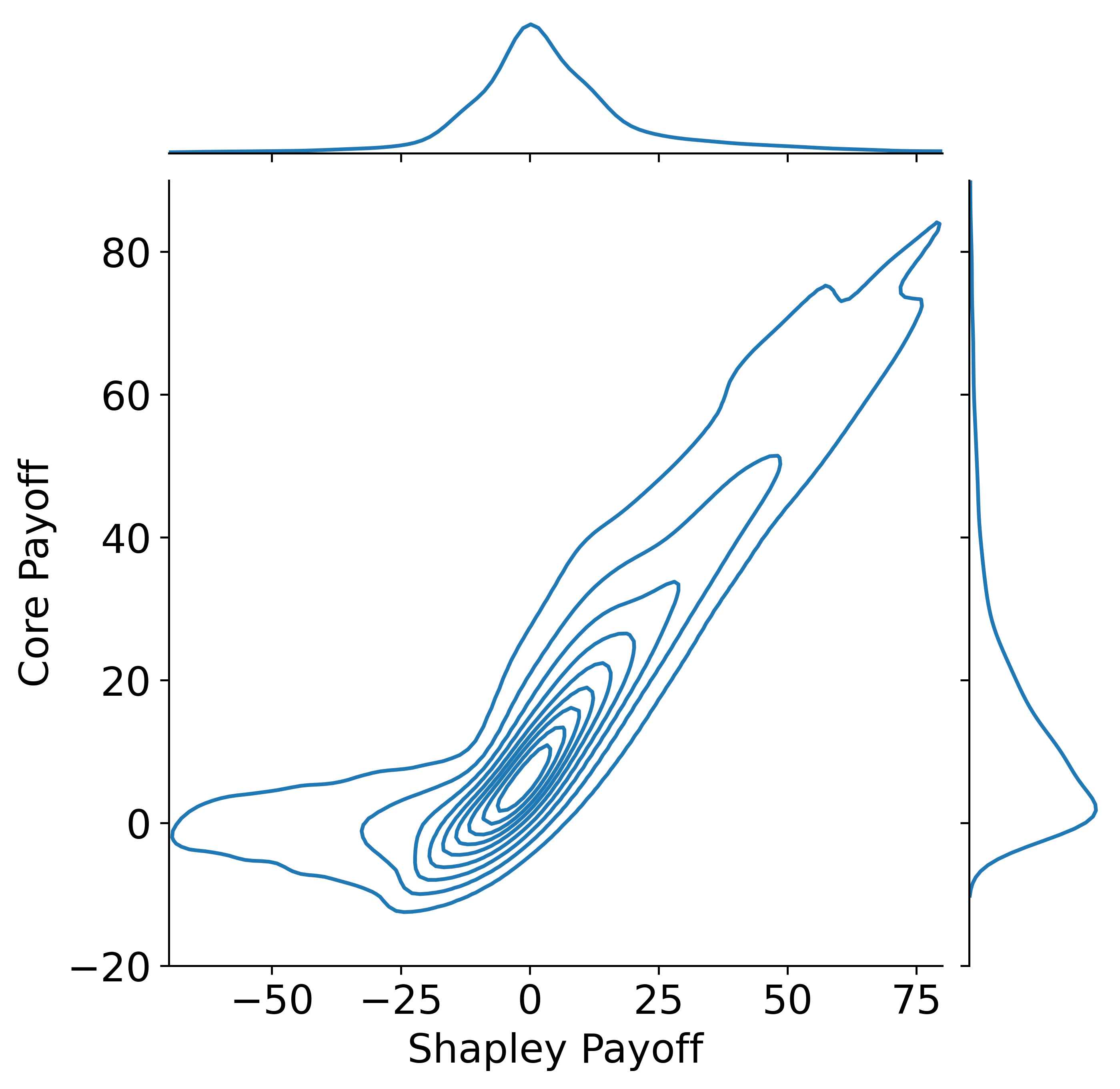}
        \label{fig:xai:boston}
    \end{subfigure}\hfill
    \end{minipage}
    \begin{minipage}[c]{0.330\textwidth}
    \centering
    \begin{subfigure}[t]{\textwidth}
        \centering
        \caption{Diabetes (global / local)}
        \includegraphics[width=0.48\textwidth]{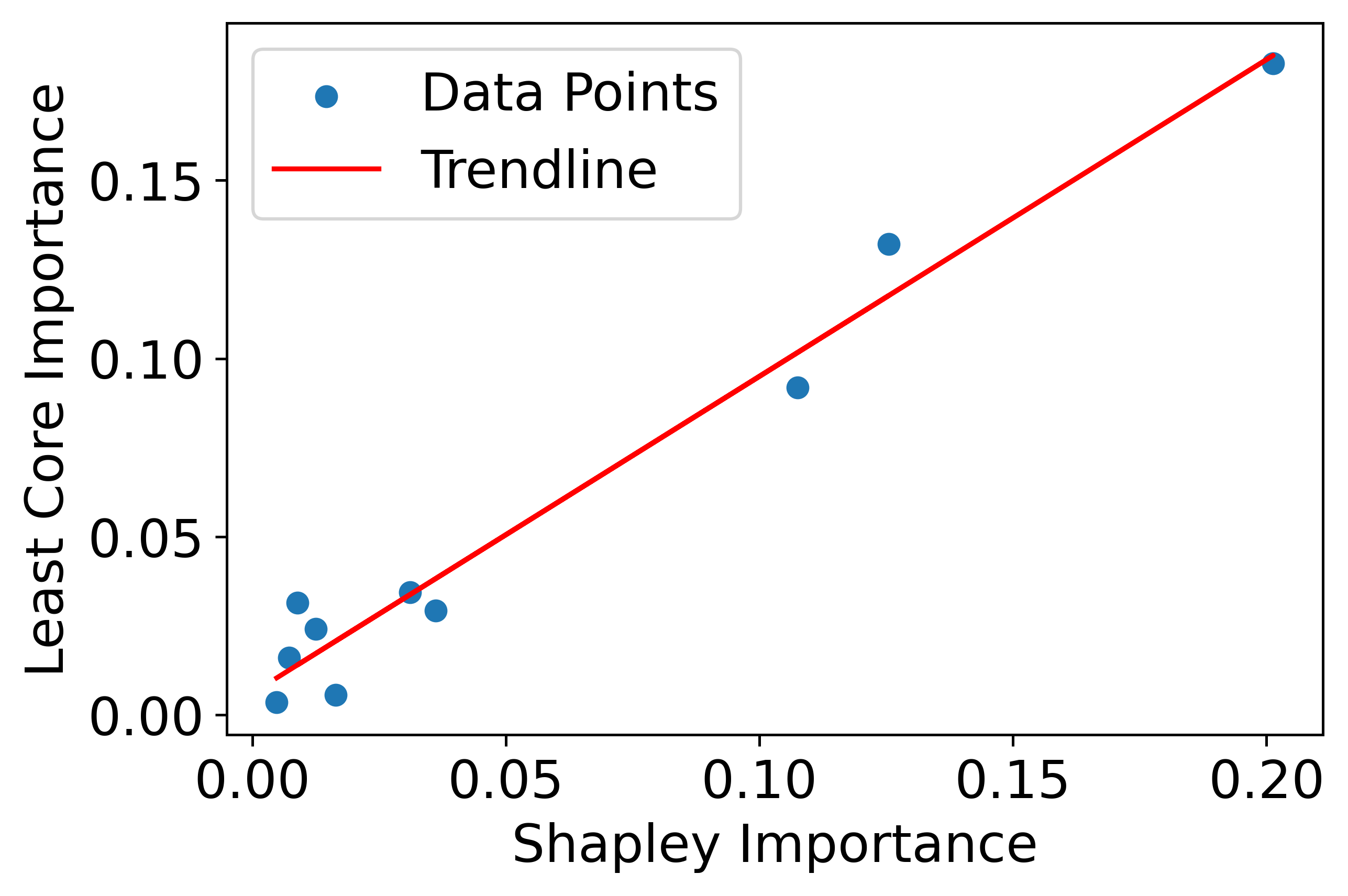}
        \includegraphics[width=0.48\textwidth]{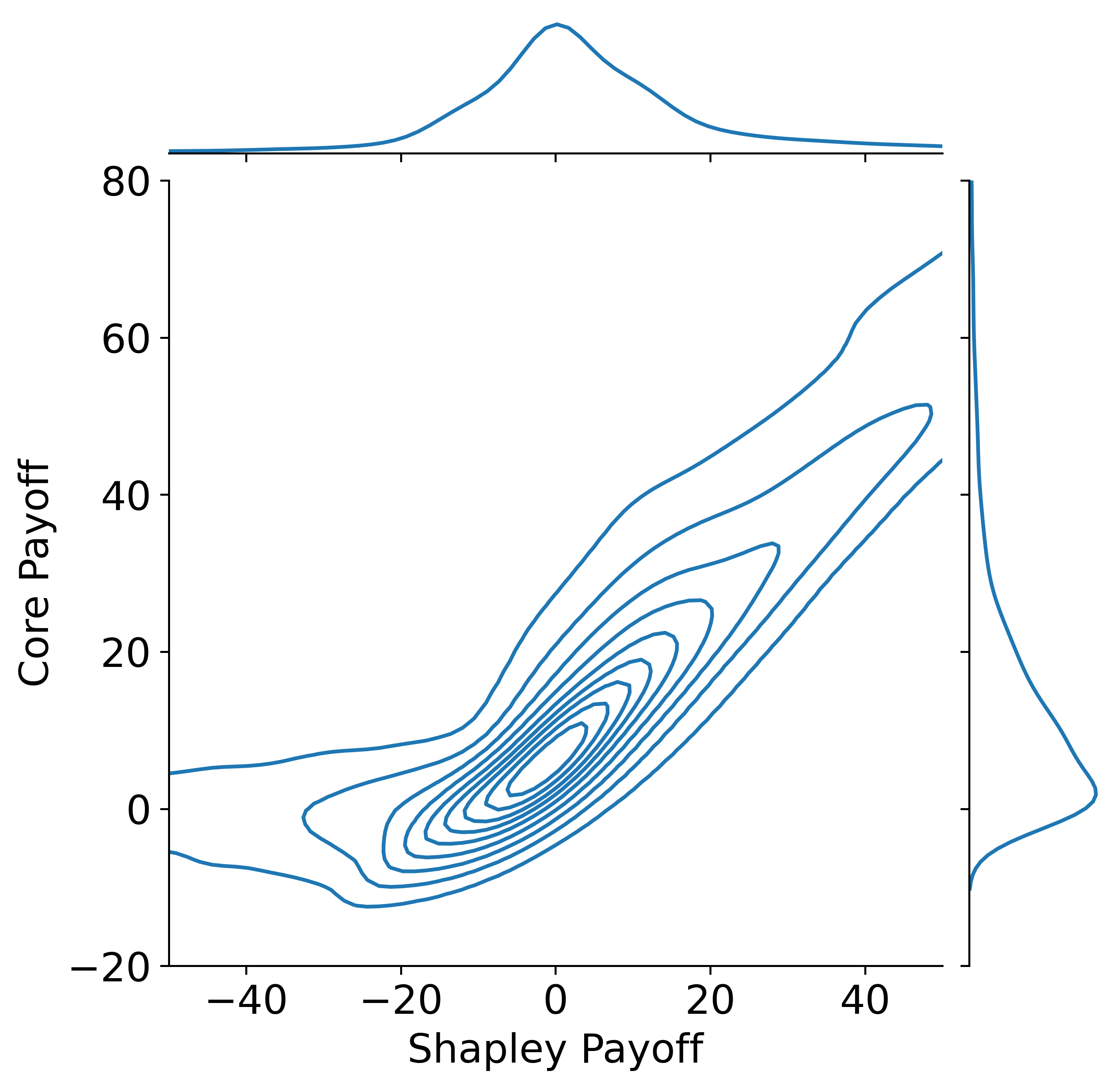}
        \label{fig:xai:diabetes}
    \end{subfigure}\hfill
    \end{minipage}
    \begin{minipage}[c]{0.330\textwidth}
    \centering
    \begin{subfigure}[t]{\textwidth}
        \centering
        \caption{Breast Cancer (global / local)}
        \includegraphics[width=0.48\textwidth]{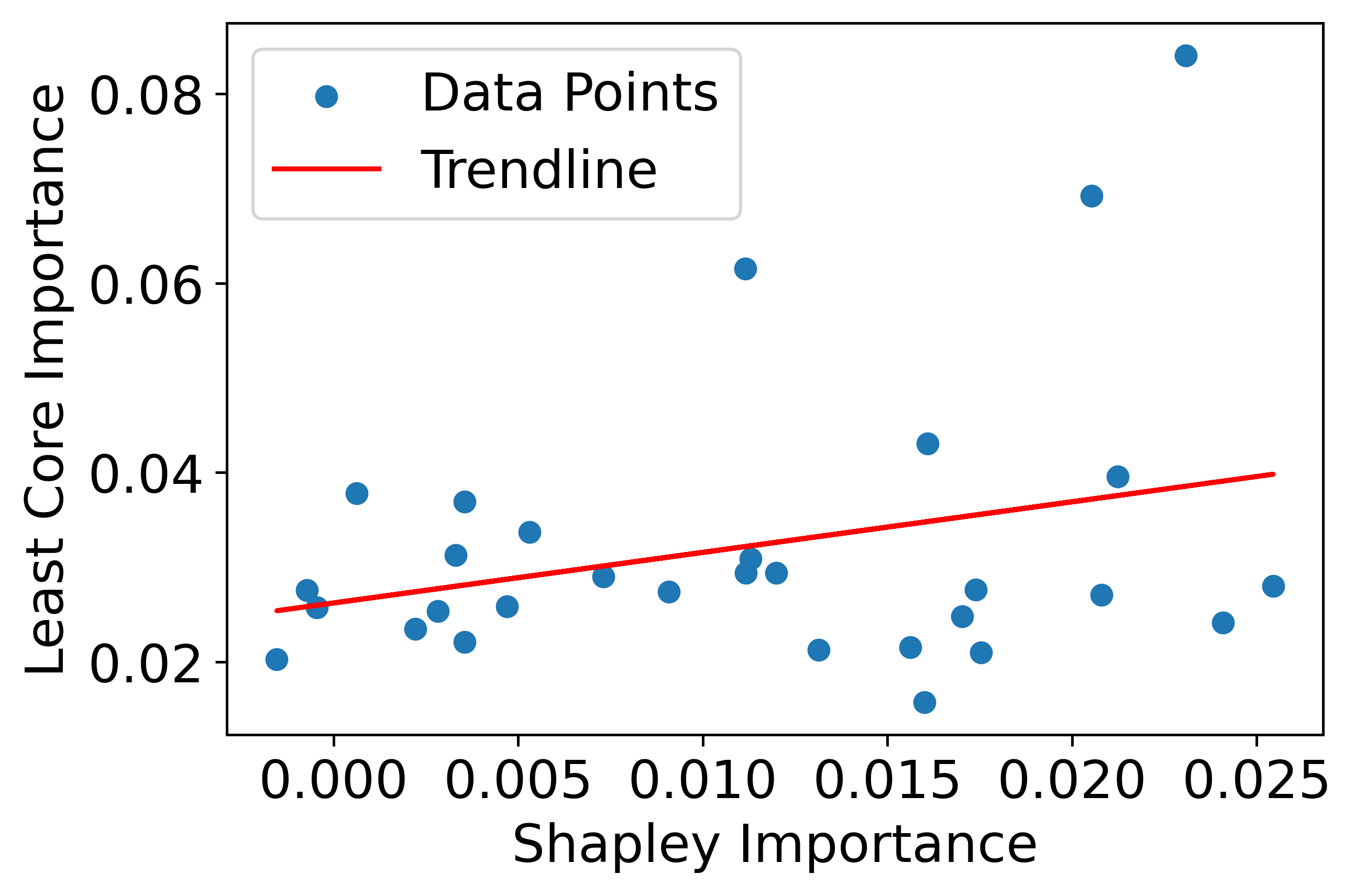}
        \includegraphics[width=0.48\textwidth]{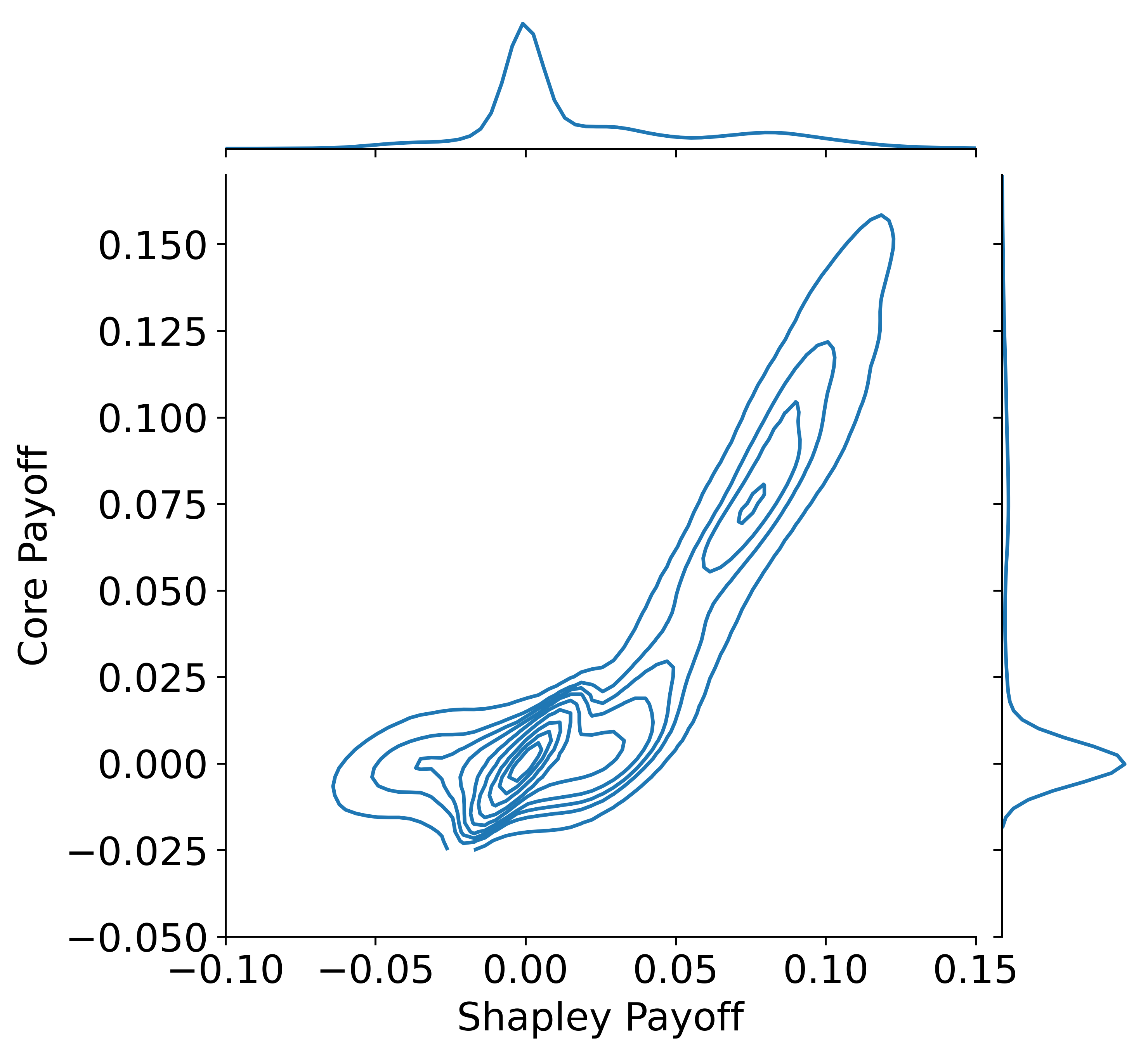} 
        \label{fig:xai:cancer}
    \end{subfigure}
    \end{minipage}
    \caption{Correlation of Shapley and Core importance measures of features at both global dataset and individual instance levels.}
    \label{fig:xai}
\end{figure*}

\begin{figure*}[ht!]
    \begin{subfigure}[t]{0.333\textwidth}
        \centering
        \caption{Boston Housing}
        \includegraphics[width=0.8\textwidth]{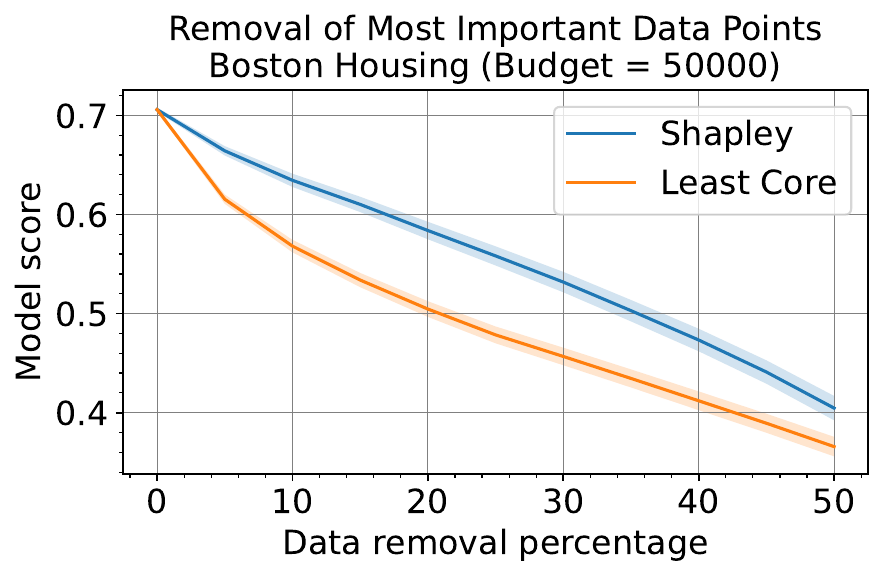}
        \label{fig:xai_data_valuation_boston_housing}
    \end{subfigure}\hfill
    \begin{subfigure}[t]{0.333\textwidth}
        \centering
        \caption{Breast Cancer}
        \includegraphics[width=0.8\textwidth]{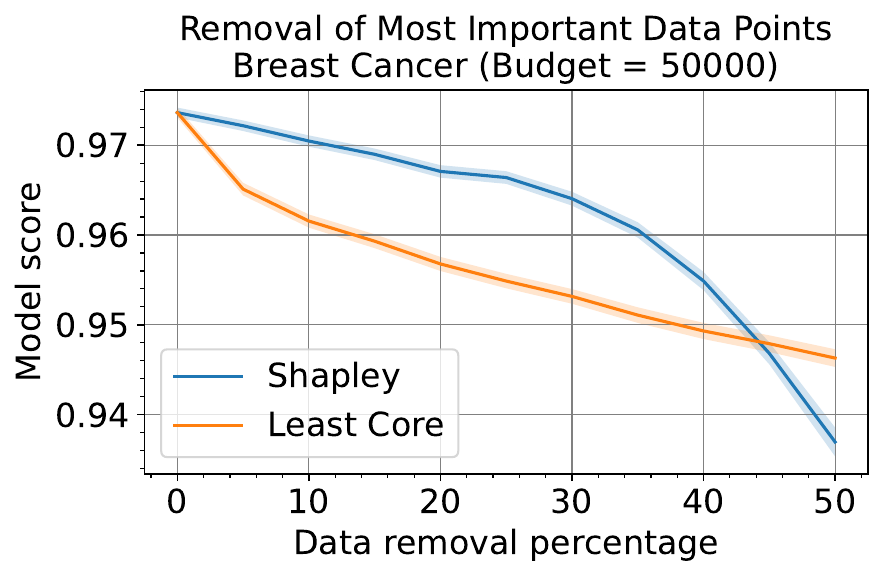} 
        \label{fig:xai_data_valuation_breast_cancer}
    \end{subfigure}\hfill
    \begin{subfigure}[t]{0.333\textwidth}
        \centering
        \caption{Chat Bot Arena}
        \includegraphics[width=0.8\textwidth]{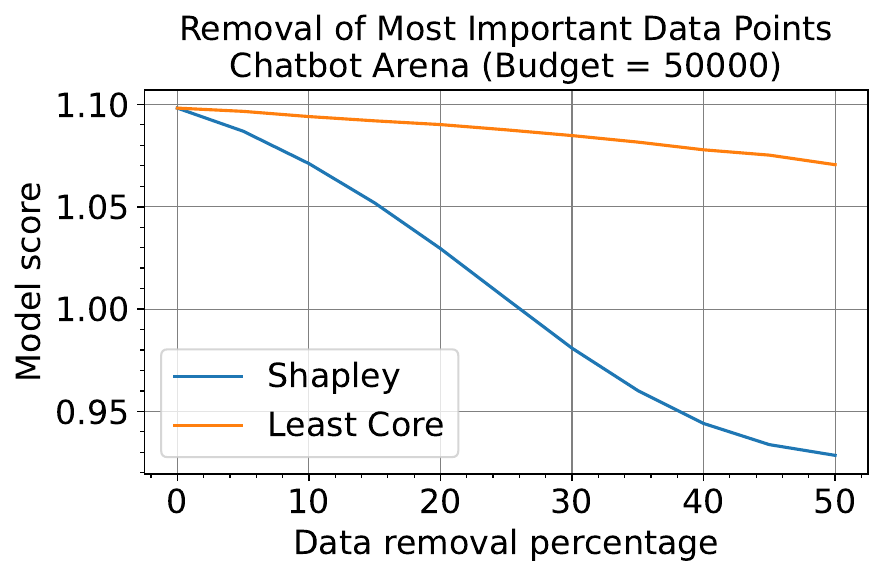}\\
        \label{fig:xai_models_appendix:arena}
    \end{subfigure}
    \caption{Data Valuation on the Boston Housing, Diabetes, and Chat Bot Arena. Error bars correspond to 95\% confidence intervals (1000 repeats).}
    \label{fig:data_valuation}
\end{figure*}

\section{Explainable AI}
\label{sect:xai_core_empirical}

Cooperative game theory has played an important role in the development of explainable AI (XAI) through the application of the Shapley value to feature-attribution and data-valuation problems~\cite{datta2016algorithmic,lundberg2017unified,antwarg2019explaining,arrieta2020explainable}. However researchers have also identified disadvantages of Shapley-based analysis, such as generating counter-intuitive explanations in various cases~\cite{sundararajan2020many}. A recent paper argued that the core may provide an alternative to the Shapley value for XAI applications. In this section we explore this idea, applying our algorithm, CL, for computing the least-core on three different real-world datasets: Boston Housing dataset~\cite{harrison1978hedonic} (price regression), Diabetes dataset~\cite{efron2004least} (classification) and the Wisconsin Breast Cancer dataset~\cite{misc_breast_cancer_wisconsin_(diagnostic)_17} (classification).  Our goal is to better understand how the Shapley value and the core are similar or different in their relative assessment of the  impact of different features or data points have on the quality of trained models, specifically, scikit-learn's default linear or logistic regression models~\citep{scikit-learn}.

\subsection{Global Explainability} 
In the global explainability problem, a full dataset is used to determine which features have high impact on the quality of the trained model. We formulate this as a cooperative game by defining a coalition, $C$, to be a subset of features and $v(C)$ to be the quality of the model trained only on features in $C$. Quality may be measured as the accuracy in classification tasks, or as the coefficient of determination $R^2$ for regression tasks. Computing the Shapley value or the least-core for this game returns individual values, one per each feature, which are can be interpreted as a measure of feature-importance.

Figure~\ref{fig:xai} presents a scatter plot showing the correlation between the feature-importance measures returned by applying the Shapley value and those produced by the least core (using our algorithm). Every point represents a single feature, with the $x$-axis reflecting its importance according to the Shapley value and the $y$-axis reflecting its importance according to the least core. 
The two importance measures are correlated, but reflect a different order of feature importance. Interestingly, the correlation is  high when there are a few dominant features (Figures~\ref{fig:xai:boston},\ref{fig:xai:diabetes}), but not as strong in the case where the predictions are driven by many features where none of the features is dominant (Figure~\ref{fig:xai:cancer}).

\subsection{Local Explainability}

In the local explainability problem, the focus is on the prediction a trained model makes on an individual instance. It seeks to determine how the values of individual features increase or decrease the model's prediction for some instance. 
We define a cooperative game where a coalition, $C$, is a set of features  and the characteristic function is the model's prediction on a new instance, created by taking some original instance from the dataset, fixing the values of all features in $C$, and replacing the values of features not in $C$ by randomly sampling other instances.\footnote{In our experiments with the least core, we use the version where non-coalition features are taken from a random instance selected from the 10\% of the instances with the lowest prediction.} Given a dataset containing $d$ instances, such an analysis results in $f$ feature importance measures per each instance, resulting in $d\cdot f$ feature-importance measures, illustrating the importance of having scalable solutions for computing both the Shapley value and least-core.

Figure~\ref{fig:xai} presents results showing the correlation between  Shapley feature-importance and least-core measures, over all the instances in each dataset. We present the data using contour plots due to the sheer number of data-points. Similar to our findings for the global explainability problem, we observe positive correlation between least-core and Shapley value. Note that the least core
returns non-negative payoffs, which explains the bend in the trend at the origin.

\subsection{Data Valuation}
\label{sec:empirical_xai:data_valuation}

In the data valuation problem, the importance of each \emph{data point} in the training set is measured. This problem can be formulated as a cooperative game by defining a coalition, $C$, to be a set of data points, with  characteristic function, $v(C)$, being the quality of the model trained only on data points in $C$. The core (or least core) is a particularly compelling solution for the data valuation problem as imputations in the core can be interpreted as prices that must be paid to data providers so as to ensure the data is available (\eg in the same way that core-pricing is used in package-auctions~\cite{day2008core,DayCramton12}).

We conducted a series of data-valuation studies across different data sets. We used a similar methodolgy as Yan and Procaccia~\cite{YanProcaccia21}: we imposed a  budget of $50,000$ calls to the characteristic function, limiting the number of permutations sampled for the Shapley value approximation and iterations for our least-core algorithm, CL. We then sorted the data by importance according to the Shapley value and least-core, and retrained the models by removing the most important data in blocks of 5\% at a time.

Results  are shown in Figure~\ref{fig:data_valuation}. In particular, these show data points deemed most important by the least-core are more critical to model performance than those selected according to Shapley value. These findings are consistent with previous literature~\cite{YanProcaccia21}. However, we also applied this process to other data sets including an evaluation problem for large language models (LLMs), presented in Appendix~\ref{app:data_valuation}. There, the results are more nuanced in that there exist scenarios where a Shapley value approach is better at identifying key data points. These results open up new research questions around characterizations as to when Shapley-based or core-based data valuation is more appropriate.

\subsubsection{Elo Ratings on Chatbot Arena data set}

We now describe another data valuation experiment that differs from the others in that it is not a traditional regression nor classification task.
The Chatbot Arena data set is composed of humans rating the quality of answers to questions to 20 different large language models (LLMs)~\cite{zheng2023judging}.
Each data point consists of a question and two answers: each answer generated by two different LLMs. The human then picks which answer is the best one and this is recorded as a win for the LLM that generated the better response and a loss for the one that generated the worse response.  In total, the data set consists of 33000 data points. 

The supervised learning problem is to learn an Elo rating~\cite{Elo78} for each LLM that can be ranked to compare the ``skill level'' (in this context: propensity to generate the better answer) of each LLM. 
Elo is classic rating system that was proposed for ranking chess engines; each LLM is assigned a rating, say $r_i$ and $r_j$, that models the probability of LLM $i$ beating LLM $j$ as a logistic function of there ratings, $\Pr(\mbox{i beats j}) = \sigma \left( (r_i - r_j)/400 \right) = \frac{1}{1 + e^{(r_j - r_i)/400}}$,
where $\sigma$ is the sigmoid function $\sigma(x) = (1 + e^{-x})^{-1}$.
The ratings can be learned using online updates, or more precisely given a batch of data using logistic regression or minorization-maximization (MM) algorithms~\cite{Hunter00}.

We apply data valuation in a similar way as before. 
For each experiment, we first sample a training set and test set pair $(\cD_R, \cD_T)$, with $|\cD_R| = 1000$ and $|\cD_T| = 10000$.
A coalition is then a subset of data points from this training set, $c \subseteq \cD_R$, leading to a 1000-player coalitional game. 
As in the other settings, the characteristic function is a measure of how well the model learned on $c$ performs on the test set $\cD_T$.
For each coalition, $c$, we run the MM method~\cite{Hunter00} for 20000 iterations to learn Elo ratings of the batch of data $c \subseteq \cD_R$. 
Given the learned ratings $\vec{r}_c$, we compute the average cross entropy loss over the test set, $L_{CE}(\vec{r}_c, \cD_T)$. 
Finally, we define the the characteristic function for as $v(c) = 2 - L_{CE}(\vec{r}_c, \cD_T)$.
As before we set a budget of 50000 calls to the characteristic function for both Shapley and the least core computations. Figure~\ref{fig:xai_models_appendix:arena} shows that in contrast to the other data sets and experiments in~\cite{YanProcaccia21}, the Shapley values attribute importances that are more critical to model performance than the least core.

\section{Conclusion}
\label{sect:conclusion}

We examined the core, a distribution of payoff over members in a coalition that is  a central solution concept in cooperative game theory.  We proposed a scalable solver for the least core that can handle the exponential number of possible coalitions that define the core constraints. We provided convergence rates and guarantees for this solver and showed empirically that it is faster than previous core solvers~\cite{YanProcaccia21}. 

Our empirical analysis shows several applications of our core solver, including studying stability of prominent forms of coalitional games and explainable AI (XAI) problems. For XAI, we highlighted that analysis based on the core differs from the current de-facto standard based on the Shapley value. Further, for the purpose of data evaluation, our results indeed show that in certain cases the core outperforms the Shapley value as a way of selecting the data instances for training machine learning models. 

Several problems remain open for future research. First, could one derive even faster algorithms for approximating the core (in terms of the worst-case performance, or in terms of empirical performance on problems of interest)? Second, could better approximation algorithms for the core be tailored to specific classes of games? Third, could one extend our analysis to other known forms of cooperative games to determine the key features that affect coalitional stability in them? Fourth, our analysis has identified ways of selecting data for training models. Could one leverage such results to speed up the training or runtime performance of machine learning models? Finally, when is it better to use the core and when is it better to use the Shapley value for feature importance measurements or data selection?

\bibliographystyle{ACM-Reference-Format} 
\bibliography{paper}

\clearpage
\appendix

\section{Algorithms}

\subsection{Proofs of Theoretical Results}
\label{app:theorem1proof}

We restate Theorem~\ref{prop:sgd-approx} and Lemma~\ref{lemma:monotone} and include their proofs.

\begin{theorem-nonumber}
If $\ell_c(\epsilon, p) \le \gamma^2$ for all $c$, then $p$ is in the $(\epsilon+\sqrt{2n}\gamma)$-core.
\end{theorem-nonumber}
\begin{proof}
\begin{align}
    \ell_c(\epsilon, p) = \frac{1}{2 \vert c \vert}\big(\max(0, v(C) - \epsilon - p^\top c)\big)^2 &\le \gamma^2 \,\, \forall c \in C
    \\ \implies \frac{1}{\sqrt{2 \vert c \vert}} (v(C) - \epsilon - p^\top c) &\le \gamma \,\, \forall c \in C
    \\ \implies v(C) - \epsilon - p^\top c &\le \sqrt{2n} \gamma \,\, \forall c \in C
    \\ \implies v(C) - (\epsilon + \sqrt{2n}\gamma) - p^\top c &\le 0 \,\, \forall c \in C.
\end{align}
\end{proof}

\begin{lemma-nonumber}\label{lemma:monontone}
The map $F$ in~\eqref{eqn:vi_map} is monotone, i.e., $\langle F(x) - F(x'), x - x' \rangle \ge 0$ over all $x, x' \in \mathcal{X}$.
\end{lemma-nonumber}
\begin{proof}
First, for any $x \in \mathcal{X}$ and $x' \in \mathcal{X}$, we find that $F(x)^\top x'$
\begin{align}
    &= \epsilon' + \mu \Big( \sum_{c \in C} \frac{d_c}{\vert C \vert} (-c^\top p' - \epsilon') \Big) - \mu' \Big( \sum_{c \in C} \ell_c(p, \epsilon) - \gamma^2 \Big)
    \\ &= \epsilon' + \mu \Big( \sum_{c \in C} \frac{d_c}{\vert C \vert} (d_c' - v(C)) \Big) - \mu' \Big( \sum_{c \in C} \ell_c(p, \epsilon) - \gamma^2 \Big).
\end{align}

Therefore, $(F(x) - F(x'))^\top (x - x')$
\begin{align}
    &= \epsilon + \mu \Big( \sum_{c \in C} \frac{d_c}{\vert C \vert} (d_c - v(C)) \Big) - \mu \Big( \sum_{c \in C} \ell_c(p, \epsilon) - \gamma^2 \Big)
    \\ &- \epsilon - \mu' \Big( \sum_{c \in C} \frac{d_c'}{\vert C \vert} (d_c - v(C)) \Big) + \mu \Big( \sum_{c \in C} \ell_c(p', \epsilon') - \gamma^2 \Big)
    \\ &- \epsilon' - \mu \Big( \sum_{c \in C} \frac{d_c}{\vert C \vert} (d_c' - v(C)) \Big) + \mu' \Big( \sum_{c \in C} \ell_c(p, \epsilon) - \gamma^2 \Big)
    \\ &+ \epsilon'+ \mu' \Big( \sum_{c \in C} \frac{d_c'}{\vert C \vert} (d_c' - v(C)) \Big) - \mu' \Big( \sum_{c \in C} \ell_c(p', \epsilon') - \gamma^2 \Big)
    \\ &= \Big( \sum_{c \in C} \frac{1}{\vert C \vert} \big[ \mu d_c (d_c - v(C)) - \mu' d_c' (d_c - v(C))
    \\ &- \mu d_c (d_c' - v(C)) + \mu' d_c' (d_c' - v(C)) \big] \Big)
    \\ &- \Big( \sum_{c \in C} (\mu' - \mu) ( \ell_c(p', \epsilon') - \ell_c(p, \epsilon) ) \Big)
    \\ &= \Big( \sum_{c \in C} \frac{1}{\vert C \vert} (d_c' - d_c)(\mu' d_c' - \mu d_c) \Big)
    \\ &- \Big( \sum_{c \in C} (\mu' - \mu) ( \ell_c(p', \epsilon') - \ell_c(p, \epsilon) ) \Big)
    \\ &= \Big( \sum_{c \in C} \frac{1}{\vert C \vert} (d_c' - d_c)(\mu' d_c' - \mu d_c) \Big)
    \\ &- \Big( \sum_{c \in C} \frac{1}{\vert C \vert} \frac{1}{2} (\mu' - \mu) ( d_c^{'2} - d_c^2) ) \Big)
    \\ &= \frac{1}{2} \Big( \sum_{c \in C} \frac{1}{\vert C \vert} (\mu + \mu') ( d_c - d_c')^2 ) \Big) \ge 0
\end{align}
confirming that $F$ is indeed monotone. 
\end{proof}

We further bound the following quantities as they may be of interest to others for further follow-up algorithmic development. We can bound the result in Lemma~\ref{lemma:monotone}. First, recognize
\begin{align}
    (d_c-d_c')^2 &= \begin{cases}
    \big( (\epsilon + c^\top p) - (\epsilon' + c^\top p') \big)^2 & \text{if $d \ge 0$ and $d' \ge 0$} \\
    \big( v(C) - \epsilon - c^\top p \big)^2 & \text{if $d \ge 0$ and $d' < 0$} \\
    \big( v(C) - \epsilon' - c^\top p' \big)^2 & \text{if $d < 0$ and $d' \ge 0$} \\
    0 & \text{else}
    \end{cases}
    \\ &\le \big( (\epsilon + c^\top p) - (\epsilon' + c^\top p') \big)^2
    \\ &+ \max \Big\{ \big( v(C) - \epsilon - c^\top p \big)^2 , \big( v(C) - \epsilon' - c^\top p' \big)^2 \Big\}
    \\ &\le \big( (\epsilon + \epsilon') + (c^\top p - c^\top p') \big)^2 + 4V(I)^2
    \\ &= ||x - x'||^2_{M} + 4V(I)^2
\end{align}
where $M = \begin{bmatrix}
        1 & c^\top & 0 \\
        c & cc^\top & 0 \\
        0 & 0 & 0
    \end{bmatrix} \preceq 2n I$.
Therefore, 
\begin{align}
    0 &\le (F(x) - F(x'))^\top (x - x')
    \\ &= \frac{1}{2} \Big( \sum_{c \in C} \frac{1}{\vert C \vert} (\mu + \mu') (d' - d)^2 ) \Big)
    \\ &\le \Big( \sum_{c \in C} \frac{1}{\vert C \vert} \overline{\mu} ( ||x - x'||^2_{M} + 4 v(I)^2 ) \Big).
\end{align}

We can also bound $||F(x) - F(x')||^2$. First, we can rewrite it as
\begin{align}
    &= || \mu' (\sum_{c \in C} \frac{d'}{C} c) - \mu (\sum_{c \in C} \frac{d'}{C} c) ||^2
    \\ &+ \Big( \mu' (\sum_{c \in C} \frac{d'}{C}) - \mu (\sum_{c \in C} \frac{d'}{C}) \Big)^2
    \\ &+ \Big( (\sum_{c \in C} \ell_c(p', \epsilon')) - (\sum_{c \in C} \ell_c(p, \epsilon)) \Big)^2
    \\ &= || \sum_{c \in C} \frac{c}{C} (\mu'd' - \mu d) ||^2
    \\ &+ \Big( \sum_{c \in C} \frac{1}{C} (\mu'd' - \mu d) \Big)^2
    \\ &+ \Big( \sum_{c \in C} \frac{1}{2C} (d^{'2} - d^2) \Big)^2
\end{align}
and then upper bound it as
\begin{align}
    \\ &\le \sum_{c \in C} || \frac{c}{C} (\mu'd' - \mu d) ||^2 + \sum_{c \in C} \Big( \frac{1}{C} (\mu'd' - \mu d) \Big)^2
    \\ &+ \sum_{c \in C} \Big( \frac{1}{2C} (d^{'2} - d^2) \Big)^2
    \\ &\le \sum_{c \in C} \frac{1}{C} (\mu'd' - \mu d)^2 + \sum_{c \in C} \Big( \frac{1}{C} (\mu'd' - \mu d) \Big)^2
    \\ &+ \sum_{c \in C} \Big( \frac{1}{2C} (d^{'2} - d^2) \Big)^2
    \\ &\le \sum_{c \in C} \frac{1}{C} \Big( 2 (\mu'd' - \mu d)^2 + \frac{1}{4} (d^{'2} - d^2)^2 \Big)
    \\ &\le \sum_{c \in C} \frac{1}{C} \Big( 2 (\mu'd' - \mu d)^2 + \frac{1}{4} (d' + d)^2 (d' - d)^2 \Big).
\end{align}

\subsection{Deployment with Accelerators}
\label{app:accelerators}

It is possible to accelerate Algorithm~\ref{alg:lagrangian_core} on GPU/TPUs. The other algorithms we proposed can reuse these techniques as well. This is a key advantage of iterative, stochastic (Monte Carlo, minibatch) algorithms. The key step to parallelize is $d_c = \max(0, v(c) - \epsilon - p^T c)$ for a batch of constraints.

Let $C$ be a $B \times n$ matrix representing a batch of $B$ coalitions with a single coalition $c$ represented on each row (we are abusing notation from earlier where $C$ represents a single coalition). Each entry of $C$ indicates whether player $i$ is included in coalition $c$. Therefore, this matrix is sparse (and binary). Let $d$ be a vector containing $d_c$ for each coalition $c$. Let $v$ be a vector containing the value of each coalition. Then in JAX~\citep{deepmind2020jax}, $d = \texttt{jax.nn.relu}(v - \epsilon - \texttt{jnp.dot}(C, p))$.

This is quite similar to the operation of processing the output of a single neural network layer followed by a \texttt{ReLU} nonlinearity. $C$ plays the role of the typical weight matrix $W$. $p$ plays the role of the typical input $x$ to a neural network layer. The offset of the neural network layer is $v - \epsilon$. And \texttt{jax.nn.relu} is the nonlinearity.

We now describe how to use the resulting output $d$ to compute the updates of Algorithm~\ref{alg:lagrangian_core}. The update in $F(x)$ (see~\eqref{eqn:vi_map}) can then be easily computed from $d$. Note $\ell(p, \epsilon) = d^2 \oslash (2 C\mathbf{1}_n)$ where $\mathbf{1}_n$ is the length-$n$ ones vector. $F(x)$ is then used to construct the \texttt{Update} in the for loop of Algorithm~\ref{alg:lagrangian_core}.

When the batch size $B$ is very large, the operation $\texttt{jnp.dot}(C, p)$ can take advantage of GPUs/TPUs to process this matrix-vector operation quickly. $C$ is sparse, which can also speed up calculations using sparse matrix multiplication libraries.

Furthermore, if the batch size $B$ is too large to fit $C$ on a single device, the batches can be split up across devices, say $d^j = \texttt{jax.nn.relu}(v^j - \epsilon - \texttt{jnp.dot}(C^j, p))$. The update in $F(x)$ can be computed similarly to before on each device (giving us $F^j(x)$) and then aggregated across devices to produce the final result.

\subsection{Algorithm Hyperparameters}
\label{app:alg_hyps}

We report the hyperparameters we used for Algorithm~\ref{alg:lagrangian_core} below. For the runtime experiment in Figure~\ref{fig:timings-wvgs}, we used those reported in Table~\ref{tab:alg_hyps:timing}.

\begin{table}[!h]
    \centering
    \begin{tabular}{|l|c|}\hline
        Learning Rate Schedule $\eta_t$ & $0.1 \rightarrow 0.01$ over $1000$ steps \\ %
        Coalition Batch Size $B$ & $100$ \\
        Lagrange Mutiplier Initial Value $\mu_{0}$ & $1000$ \\
        Constraint Violation Threshold $\gamma$ & $0.001$ \\
        Number of Iterations $T$ & $10,000$ \\ \hline
    \end{tabular}
    \caption{Algorithm~\ref{alg:lagrangian_core} Hyperparameters: Runtime Experiment. Alternatively, the linear schedule can be written as $\eta_t = \max \big(0.01, 0.1 \cdot (1 - \frac{t}{1000}) + 0.01 (\frac{t}{1000}) \big)$.}
    \label{tab:alg_hyps:timing}
\end{table}

For the remaining experiments, we used a change of variables and replaced $p$ with $\softmax(s)$ where $s$, our new decision variable, is a vector of $n$ real-valued logits. This then allowed us to use the more sophisticated update operator Adam~\citep{kingma2014adam}. Note that the JAX~\citep{deepmind2020jax} autodiff library makes it trivial to compute $F_{s} = \nabla_s \mathcal{L}(p(s), \epsilon, \mu)$ after the change of variables. We used the hyperparameters reported in Table~\ref{tab:alg_hyps:no_timing}.

\begin{table}[!h]
    \centering
    \begin{tabular}{|l|c|}\hline
        Learning Rate Schedule $\eta_t$ & $0.1$ \\
        Coalition Batch Size $B$ & $1000$ \\
        Lagrange Mutiplier Initial Value $\mu_{0}$ & $1000$ \\
        Constraint Violation Threshold $\gamma$ & $0.01$ \\
        Number of Iterations $T$ & $10,000$ \\ \hline
    \end{tabular}
    \caption{Algorithm~\ref{alg:lagrangian_core} Hyperparameters: Other Experiments}
    \label{tab:alg_hyps:no_timing}
\end{table}

\section{Shapley Value versus the Core}
\label{sec:shapley_vs_core}

The Shapley value is an intuitive solution concept for quantifying the expected marginal contribution of a player. The {\it marginal contribution} of agent $i$ to a coalition $C$ such that $i \notin C$ is the increase in utility the coalition $C$ achieves when $i$ joins it, defined as $m_i^C = v(C\cup\{i\})-v(C)$. We can define a similar concept for {\it permutations} of agents.  We denote the predecessors of $i$ in the the permutation $\pi$ as $b(i,\pi)$, i.e. the agents appearing before $i$ in the permutation $\pi$. We denote the set of all permutations over the agents as $\Pi$. The  marginal contribution of player $i$ in the permutation $\pi$ is the increase in utility $i$ provides when joining the team of agents appearing before them in the permutation,  $m_i^\pi=v(b(i,\pi)\cup \{i\})-v(b(i,\pi))$. The {\bf Shapley value} of agent $i$ is defined as their marginal contribution, averaged across all permutations:

\begin{equation} \label{eq:shapley}
	\Phi_i(v) = \frac{1}{N!}\sum_{\pi\in\Pi} v(b(i,\pi))\cup\{i\})-v (b(i,\pi)) 
\end{equation}

The Shapley value uniquely exhibits four axiomatic properties~\cite{Chalkiadakis12Computational}: {\it efficiency}: $\sum_{i=1}^n \Phi_i(v) = v(I)$,
{\it null player}: if $i$ contributes nothing, then $\Phi_i(v) = 0$, {\it symmetry}: $\forall C \subseteq I - \{i,j\}, v(C \cup \{i\}) = v(C \cup \{j\}) \Rightarrow \Phi_i(v) = \Phi_j(v)$, and {\it additivity}: for two games $v_1, v_2$, the Shapley value of the combined game $\Phi_i(v_1 + v_2) = \Phi_i(v_1) + \Phi_i(v_2)$.
One advantage of the Shapley value is that it can also be computed using Monte Carlo sampling~\cite{bachrach2010approximating,Mitchell22}. The Shapley value has become widely-used as a means of quantifying the effects of individual features in XAI~\cite{lundberg2017unified}.

As discussed in the main paper, there are some computational and conceptual limitations of the Shapley value. Computationally, there are some games where no polynomial-time algorithm can obtain sufficiently small confidence intervals~\cite{Bachrach08Coalitional}
and others for which it cannot be computed using a uniform distribution~\cite{Balkanski17Statistical}.
For XAI applications, the Shapley value has been criticized for several reasons. Firstly, the Shapley value does not necessarily reflect what human end users really want to explain about models in deployment~\cite{Bhatt19}. As such, it has been labeled ``not a natural solution to the human-centric issues of explainability''~\cite{Kumar20}.
Second, when features are correlated, it is not clear how Shapley values should attribute relative importance~\cite{Kumar20}.
Third, the usefulness of the additivity axiom is questionable~\cite{YanProcaccia21} and can be uninformative for non-additive models~\cite{Kumar20}.
Finally, when humans aided with a model were evaluated on alert processing, there was no significant difference in task utility metrics when the Shapley values were provided versus  not provided~\cite{Weerts19}.

In contrast, the core focuses on the {\it stability} of the imputations. 
In particular, utility subdivisions discovered by the core can be interpreted as ``economically plausible payments in a competitive market in the sense that every coalition should be compensated for its market value''~\cite{YanProcaccia21}.
Imputations outside the core would lead to some agents having incentive to form different coalitions for higher value and redistribute the extra payoff among themselves. In some contexts, approximations of the least core can attribute importance more reliably than Shapley value in XAI~\cite{YanProcaccia21}.
Outside of XAI, the core is used to select payment rules in combinatorial auctions due to problems with the commonly-used mechanisms~\cite{Moor16,DayCramton12}.

Lastly, there is evidence that the Shapley value has weaker empirical support than the core for predicting human payment divisions in behavioral studies~\cite{Williams88}.
In a more recent study, behavior of players deviated from the Shapley value predictions and violated two of its axioms~\cite{DEonLarson20}: null player and additivity. 
For all these reasons, we believe that the least core is a worthy alternative solution concept in general and also for XAI.

\subsection{Additional Timing Experiment}\label{app:timing}
In Figure~\ref{fig:timings-wggs}, we replicate the weighted voting games timing experiment from Figure~\ref{fig:timings-wvgs} but now for a weighted graph game, specifically the Newman Watts Strogatz game in Figure~\ref{fig:wgg:nws} with game parameters $p$ and $k$ sampled from $[4, 24]$ and $[0, 1]$ as indicated in the figure. We also display two additional runs of the Core Lagrangian approach with two different learning rates (LR) around the base learning rate: learning rate schedule $= \texttt{linear\_schedule}(2 \times LR, 0.1 \times LR, 10^3)$ with base $LR = 0.5$. Lastly, we report approximation quality for the Monte-Carlo Shapley valueas measured by euclidean distance to its approximate ground truth value. In this case, the LP method outperforms the CL method suggesting it is important keep both approaches in mind when attempting to solve new problem domains.

\begin{figure}[t]
\centering
\includegraphics[width=0.48\textwidth]{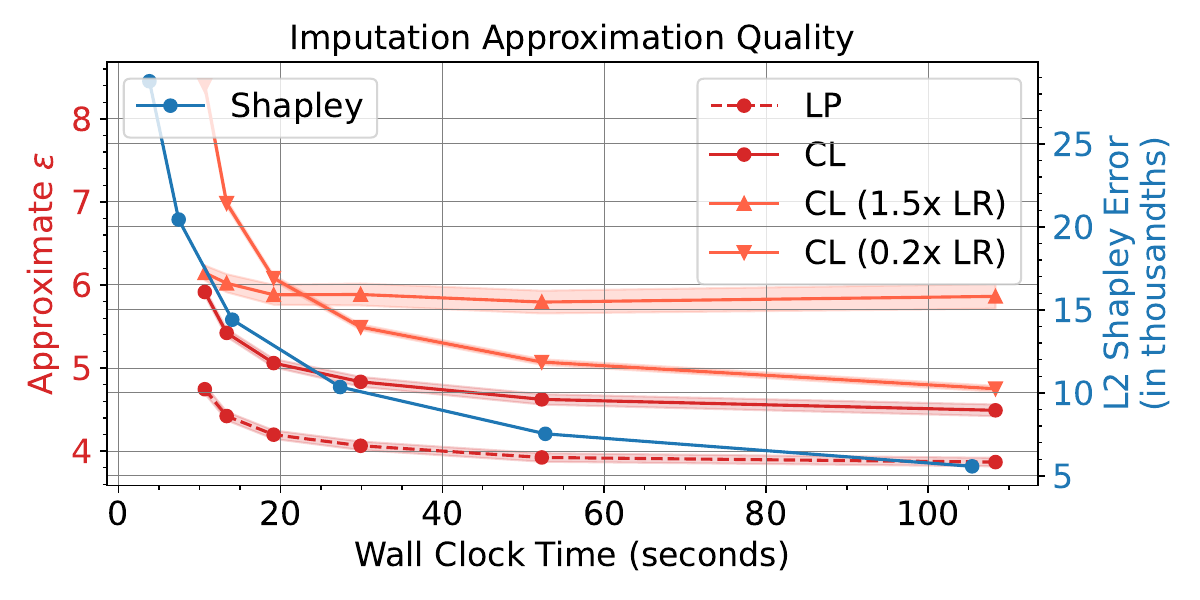}
\caption{Approximation quality of the linear program (LP), core Lagrangian (CL), and Shapley method as a function of computation time (seconds).
The $x$-axis corresponds to wall-clock time taken by all algorithms run side-by-side given $k \in \{500, 1000, 2000, 4000, 8000, 16000\}$ sampled coalitions for the LP method, $t_k$. A better core approximation quality is reflected in having a lower $\epsilon$ for the same runtime. The left $y$-axis represents the approximate $\epsilon$ of the least-core solution found, $\hat{\epsilon}(p_{LP}, \hat{C})$ and  $\hat{\epsilon}(p_{CL}, \hat{C})$, computed over the same set of $2^{25}$ coalitions, $\hat{C}$.
As the Shapley solution is not motivated by stability, we measure the approximate solution's quality by euclidean distance to the Shapley value estimated with $10$ million Monte-Carlo samples. Each Shapley approximation is estimated using $125 \times k$ samples.
Each data point $(t_k, \hat{\epsilon}(p_x, \hat{C}))$ represents an average over the same set of $2,500$ random weighted voting games with shading indicating standard error of the mean.
\label{fig:timings-wggs}}
\end{figure}

\section{Graph Games}
\label{app:graph-games}

\subsection{Other graph games}

We present $4$ more graph games to complement Section~\ref{sec:graph-games}:

{\bf Partition Graph} are graphs where the vertices are split into $n$ partitions. Similar to the Erdős-Rényi graph, an edge occurs between vertices {\em within} a partition with probability $p_{in}$ and with probability $p_{out}$ {\em across} partitions.
{\bf Dual Barabasi Albert} graphs are generated by connecting each new vertex with either $m_1$ edges with probability $p$ or $m_2$ edges with probability $1-p$ \citep{moshiri2018dualbarabasialbert} (with edge targets selected at random from the current vertices). 
{\bf Powerlaw Clusters} are graphs with a powerlaw degree distribution and approximate average clustering \citep{holme2002powerlawgraph}, where $m$ is the number of edges to randomly add to each node and $p$ is the probability of adding a triangle after adding each edge.
{\bf Random Uniform Intersection} A random subset of $m$ elements is assigned to every vertex in a graph, each element being assigned independently with probability $p$. An edge between two vertices exists if their element subsets intersect.

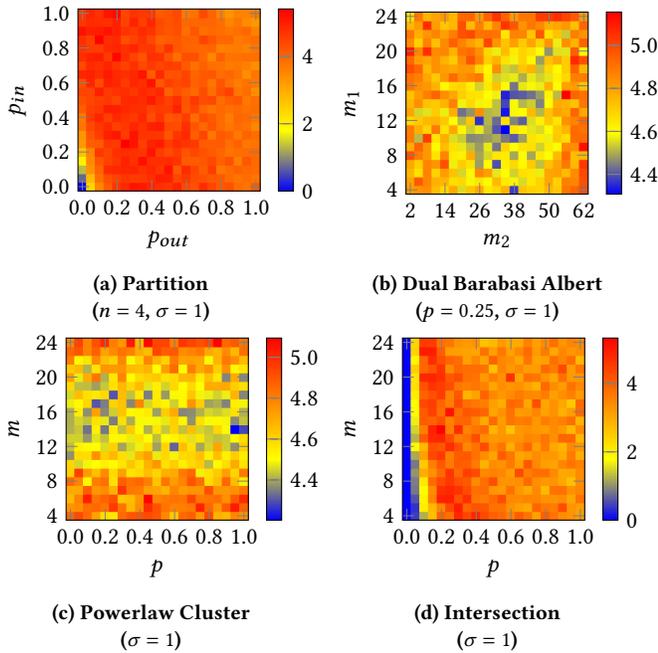
\begin{figure}[t]
    \centering
    
    \begin{subfigure}{0.225\textwidth}
        \centering
        \begin{tikzpicture}
            \begin{axis}[
                xlabel=$p_{out}$,
                ylabel=$p_{in}$,
                ylabel near ticks,
                xlabel near ticks,
                xtick={0,4,8,12,16,20},
                xticklabels={0.0,0.2,0.4,0.6,0.8,1.0},
                ytick={0,4,8,12,16,20},
                yticklabels={0.0,0.2,0.4,0.6,0.8,1.0},
                enlargelimits=false,
                axis on top,
                colorbar,
                colorbar style={
                    xlabel near ticks,
                    yticklabel style={
                        /pgf/number format/.cd,
                        fixed,
                        precision=0,
                        fixed zerofill,
                    },
                    width=2mm,
                    at={(1.1,1)},
                },
                width=\linewidth,
                height=\linewidth,
            ]
                \addplot [matrix plot*,point meta=explicit] file [] {data/sweep_partition.dat};
            \end{axis}
        \end{tikzpicture}
        \captionsetup{justification=centering}
        \caption{Partition\\($n=4$, $\sigma=1$)}
    \end{subfigure}\hfill%
    \begin{subfigure}{0.225\textwidth}
        \centering
        \begin{tikzpicture}
            \begin{axis}[
                xlabel=$m_2$,
                ylabel=$m_1$,
                ylabel near ticks,
                xlabel near ticks,
                xtick={0,4,8,12,16,20},
                xticklabels={2,14,26,38,50,62},
                ytick={0,4,8,12,16,20},
                yticklabels={4,8,12,16,20,24},
                enlargelimits=false,
                axis on top,
                colorbar,
                colorbar style={
                    xlabel near ticks,
                    yticklabel style={
                        /pgf/number format/.cd,
                        fixed,
                        precision=1,
                        fixed zerofill,
                    },
                    width=2mm,
                    at={(1.1,1)},
                },
                width=\linewidth,
                height=\linewidth,
            ]
                \addplot [matrix plot*,point meta=explicit] file [] {data/sweep_dual_barabasi.dat};
            \end{axis}
        \end{tikzpicture}
        \captionsetup{justification=centering}
        \caption{Dual Barabasi Albert\\($p=0.25$, $\sigma=1$)}
    \end{subfigure}%
    
    \begin{subfigure}{0.225\textwidth}
        \centering
        \begin{tikzpicture}
            \begin{axis}[
                xlabel=$p$,
                ylabel=$m$,
                ylabel near ticks,
                xlabel near ticks,
                xtick={0,4,8,12,16,20},
                xticklabels={0.0,0.2,0.4,0.6,0.8,1.0},
                ytick={0,4,8,12,16,20},
                yticklabels={4,8,12,16,20,24},
                enlargelimits=false,
                axis on top,
                colorbar,
                colorbar style={
                    xlabel near ticks,
                    yticklabel style={
                        /pgf/number format/.cd,
                        fixed,
                        precision=1,
                        fixed zerofill,
                    },
                    width=2mm,
                    at={(1.1,1)},
                },
                width=\linewidth,
                height=\linewidth,
            ]
                \addplot[matrix plot*,point meta=explicit] file [] {data/sweep_powerlaw_cluster.dat};
            \end{axis}
        \end{tikzpicture}
        \captionsetup{justification=centering}
        \caption{Powerlaw Cluster\\($\sigma=1$)}
    \end{subfigure}\hfill%
    \begin{subfigure}{0.225\textwidth}
        \centering
        \begin{tikzpicture}
            \begin{axis}[
                xlabel=$p$,
                ylabel=$m$,
                ylabel near ticks,
                xlabel near ticks,
                xtick={0,4,8,12,16,20},
                xticklabels={0.0,0.2,0.4,0.6,0.8,1.0},
                ytick={0,4,8,12,16,20},
                yticklabels={4,8,12,16,20,24},
                enlargelimits=false,
                axis on top,
                colorbar,
                colorbar style={
                    xlabel near ticks,
                    yticklabel style={
                        /pgf/number format/.cd,
                        fixed,
                        precision=0,
                        fixed zerofill,
                    },
                    width=2mm,
                    at={(1.1,1)},
                },
                width=\linewidth,
                height=\linewidth,
            ]
                \addplot [matrix plot*,point meta=explicit] file [] {data/sweep_intersection.dat};
            \end{axis}
        \end{tikzpicture}
        \captionsetup{justification=centering}
        \caption{Intersection\\($\sigma=1$)}
    \end{subfigure}

    \caption{Mean LCV for different classes of graphs when sweeping over two hyperparameters. Constant hyperparameters are shown in parenthesis.}
    \label{fig:sweep_graphs_appendix}
\end{figure}

We observe in Figure~\ref{fig:sweep_graphs_appendix} that similarly to Erdős-Rényi, partition graphs have the opposite phenomenon of a small corner of the space with stable games (low $p_{in}$ and $p_{out}$), with the games with moderate values of the two parameters being less stable. In contrast, Dual Barabasi Albert graphs have stable games when both parameters have medium values, with less stability as either of the parameters deviating from these values. Finally, similarly to Newman Watts Strogatz, Intersection graphs indicate a spectrum between stable and unstable games (with the $p$ parameter having strong influence in Intersection graphs, and $k$ parameter having the stronger influence on stability in NWS graphs).

\subsection{Example Graphs}

Figure~\ref{fig:example_graph_game} presents an example of an induced subgraph game, describing how the charactersitic function is calculated in this game. 

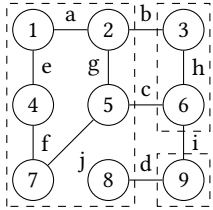
\begin{figure}[!ht]
    \centering
    \begin{tikzpicture}[scale=1.0]
        \node[circle, draw] (1) at (0,2) {1};
        \node[circle, draw] (2) at (1,2) {2};
        \node[circle, draw] (3) at (2,2) {3};
        \node[circle, draw] (4) at (0,1) {4};
        \node[circle, draw] (5) at (1,1) {5};
        \node[circle, draw] (6) at (2,1) {6};
        \node[circle, draw] (7) at (0,0) {7};
        \node[circle, draw] (8) at (1,0) {8};
        \node[circle, draw] (9) at (2,0) {9};
        
        \draw [dashed] ++(-0.35,-0.35) rectangle ++(1.7,2.7);
        \draw [dashed] ++(1.65,0.65) rectangle ++(0.7,1.7);
        \draw [dashed] ++(1.65,-0.35) rectangle ++(0.7,0.7);
        
        \graph {
            (1) --[edge label=a] (2) --[edge label=b] (3);
                   (5) --[edge label=c] (6);
                   (8) --[edge label=d] (9);
            
            (1) --[edge label=e] (4) --[edge label=f] (7);
            (2) --[edge label'=g] (5);
            (3) --[edge label=h] (6);
            (6) --[edge label=i] (9);

            (7) --[edge label'=j] (5);
        };
    \end{tikzpicture}
    \caption{Example graph game with 9 players. The coalition with players $\{1, 2, 4, 5, 7, 8\}$ has payoff $a+e+g+f+j$, the coalition with players $\{3, 6\}$ has payoff $h$, and the coalition with players $\{ 9 \}$ has payoff $0$ (as do all singleton coalitions). Intuitively, the edges correspond to synergies between players that are only realised if they are in the same coalition.}
    \label{fig:example_graph_game}
\end{figure}

Figures \ref{fig:example_graphs} show examples of randomly generated graphs, of the various classes considered in the paper, along with their representation as weighted adjacency matrices. 

\subsubsection{Marginal Contribution Networks}
\label{sect:empirical_mcn}

The final class of games we study are the marginal contribution networks (MCNs)~\cite{ieong2005marginal}.
In an MCN, the cooperative game is defined by a set of rules $\{r_1,\ldots, r_k\}$ such that each rule takes the form $r=(P,N)$ where $P,N\subseteq I$ refer to the positive and negative parts respectively.  Given a coalition $C\subseteq I$, a rule $r=(P,N)$ \emph{applies} to $C$ if and only if
$P\subseteq C$ and $C\cap N=\emptyset$. Clearly multiple rules can apply to a coalition and so we let $R_C$ be the set of rules that apply to $C$. The resulting characteristic function is thus $v(C) = \sum_{r \in R_C} w(r)$.

To generate random MCNs, we generate $k$ rules, each using the following sampling process: for each player, they are included in $P$ with probability $p$, and in $N$ with probability $q$.  If an agent is selected for both $P$ and $N$ the rule is discarded and re-sampled.

Figure~\ref{fig:mcn_lcv} shows the impact of the parameters on the expected least-core value. In Figure~\ref{fig:mcn_lcv_a} we vary $p$ and $q$ while in Figure~\ref{fig:mcn_lcv_b} we change both the number of rules and the number of players. 
In particular, we observe that the expected least-core value (and hence, the expected stability of the game) depends on all four parameters. For example, games with low $p$ but high $q$ (i.e. coalitions are more likely to not have any rules apply to them) are very stable, but so are games with few rules and many agents.  These insights open up new research questions such as the possibility for designing special-purpose algorithms for MCNs under specific parameter-regimes.

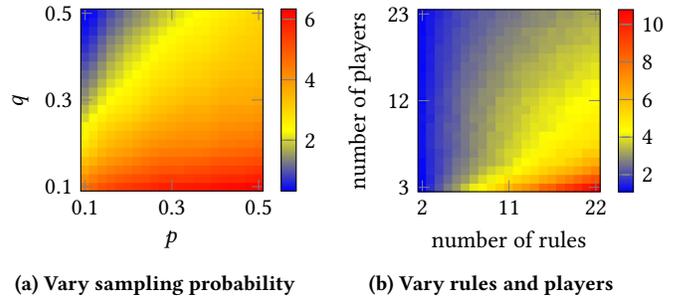
\begin{figure}[t]%
    \begin{subfigure}[t]{0.225\textwidth}
        \centering
        \begin{tikzpicture}
            \begin{axis}[
                xlabel=$p$,
                ylabel=$q$,
                ylabel near ticks,
                xlabel near ticks,
                xtick={0,10,20},
                xticklabels={0.1,0.3,0.5},
                ytick={0,10,20},
                yticklabels={0.1,0.3,0.5},
                enlargelimits=false,
                axis on top,
                colorbar,
                colorbar style={
                    xlabel near ticks,
                    yticklabel style={
                        /pgf/number format/.cd,
                        fixed,
                        precision=0,
                        fixed zerofill,
                    },
                    width=2mm,
                    at={(1.1,1)},
                },
                width=\linewidth,
                height=\linewidth,
            ]
                \addplot [matrix plot*,point meta=explicit] file [] {data/mcn_p_vs_q.dat};
            \end{axis}
        \end{tikzpicture}
        \caption{Vary sampling probability}
        \label{fig:mcn_lcv_a}
    \end{subfigure}\hfill
    \begin{subfigure}[t]{0.225\textwidth}
        \centering
        \begin{tikzpicture}
            \begin{axis}[
                xlabel=number of rules,
                ylabel=number of players,
                ylabel near ticks,
                xlabel near ticks,
                xtick={0,10,20},
                xticklabels={2,11,22},
                ytick={0,10,20},
                yticklabels={3,12,23},
                enlargelimits=false,
                axis on top,
                colorbar,
                colorbar style={
                    xlabel near ticks,
                    yticklabel style={
                        /pgf/number format/.cd,
                        fixed,
                        precision=0,
                        fixed zerofill,
                    },
                    width=2mm,
                    at={(1.1,1)},
                },
                width=\linewidth,
                height=\linewidth,
            ]
                \addplot [matrix plot*,point meta=explicit] file [] {data/mcn_player_vs_rule.dat};
            \end{axis}
        \end{tikzpicture}
        \caption{Vary rules and players}
        \label{fig:mcn_lcv_b}
    \end{subfigure}
    \caption{Mean least-core value over $200$ randomly generated MCNs ($10$ players, $k=10$ rules).}
    \label{fig:mcn_lcv}
\end{figure}

\section{XAI: Datasets and Additional Experiments}

\subsection{Datasets}
\label{app:xai_datasets}

We list statistics of our XAI datasets in Table~\ref{tab:xai_datasets}. We used 80\% of the data for training and kept 20\% aside for testing and measuring the performance of the AI model.

\begin{table}[]
    \centering
    \begin{tabular}{c|c|c|c}
        Dataset & \# of Features & \# of Samples & Task \\ \hline
        Breast Cancer & 569 & 30 & Classification \\
        Diabetes & 442 & 10 & Regression \\
        Wine & 178 & 13 & Classification \\
        Boston Housing & 506 & 13 & Regression
    \end{tabular}
    \caption{Statistics of ML Tasks in XAI Experiments}
    \label{tab:xai_datasets}
\end{table}

\subsection{Additional Data Valuation Experiments}
\label{app:data_valuation}

We repeated the data valuation experiment described in Section~\ref{sect:xai_core_empirical} on all the datasets presented in the global and local experiments. 
Recall that data valuation assigns importances to each data point rather than features. 
Figure~\ref{fig:xai_models_appendix} shows a similar trend that removing the top data points the least core seems to be more critical for model performance than Shapley values. The trend sometimes switches after 40-50\% of data removed.

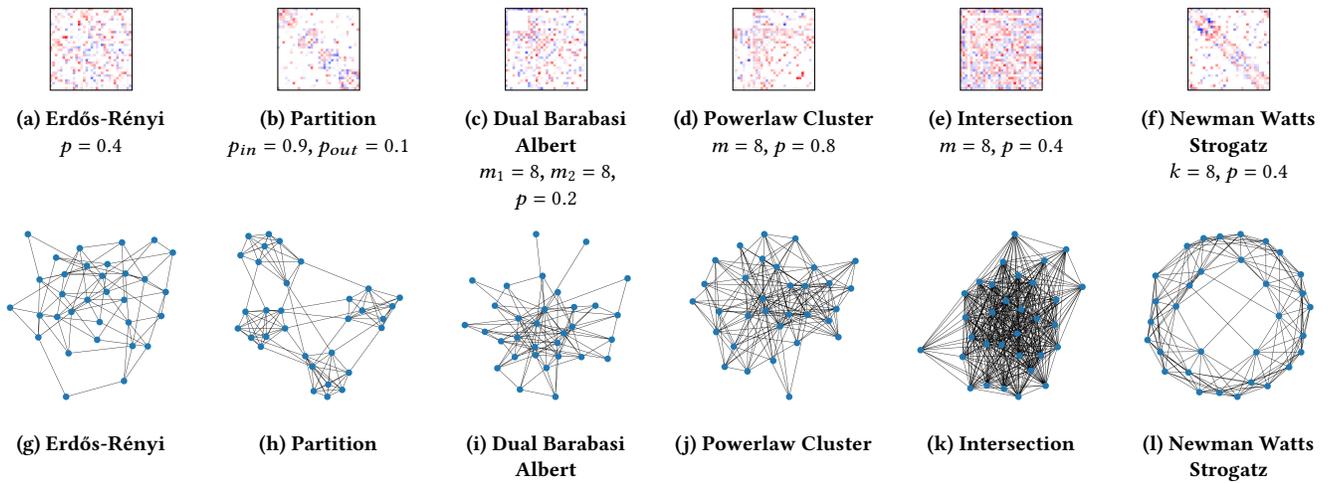
\begin{figure*}[t!]
    \centering
    \pgfplotsset{
        colormap={redyellowgreen}{rgb255=(255,0,0) rgb255=(255,255,155) rgb255=(0,128,0)}
    }
    \pgfplotsset{
        colormap={bluewhitered}{rgb255=(0,0,255) rgb255=(255,255,255) rgb255=(255,0,0)}
    }
    \begin{subfigure}[t]{0.15\textwidth}
        \centering
        \begin{tikzpicture}%
            \begin{axis}[
                colormap name=bluewhitered, 
                point meta min=-2,
                point meta max=+2,
                enlargelimits=false,
                axis on top,
                width=\linewidth,
                height=\linewidth,
                ticks=none,
            ]
                \addplot [matrix plot,point meta=explicit] file [] {data/example_erdos_renyi.dat};
            \end{axis}
        \end{tikzpicture}
        \captionsetup{justification=centering}
        \caption{Erdős-Rényi\\$p=0.4$}
    \end{subfigure}\hfill%
    \begin{subfigure}[t]{0.15\textwidth}
        \centering
        \begin{tikzpicture}%
            \begin{axis}[
                colormap name=bluewhitered, 
                point meta min=-2,
                point meta max=+2,
                enlargelimits=false,
                axis on top,
                width=\linewidth,
                height=\linewidth,
                ticks=none,
            ]
                \addplot [matrix plot,point meta=explicit] file [] {data/example_partition.dat};
            \end{axis}
        \end{tikzpicture}
        \captionsetup{justification=centering}
        \caption{Partition\\$p_{in}=0.9$, $p_{out}=0.1$}
    \end{subfigure}\hfill%
    \begin{subfigure}[t]{0.15\textwidth}
        \centering
        \begin{tikzpicture}%
            \begin{axis}[
                colormap name=bluewhitered, 
                point meta min=-2,
                point meta max=+2,
                enlargelimits=false,
                axis on top,
                width=\linewidth,
                height=\linewidth,
                ticks=none,
            ]
                \addplot [matrix plot,point meta=explicit] file [] {data/example_dual_barabasi.dat};
            \end{axis}
        \end{tikzpicture}
        \captionsetup{justification=centering}
        \caption{Dual Barabasi Albert\\$m_1=8$, $m_2=8$, $p=0.2$}
    \end{subfigure}\hfill%
    \begin{subfigure}[t]{0.15\textwidth}
        \centering
        \begin{tikzpicture}%
            \begin{axis}[
                colormap name=bluewhitered, 
                point meta min=-2,
                point meta max=+2,
                enlargelimits=false,
                axis on top,
                width=\linewidth,
                height=\linewidth,
                ticks=none,
            ]
                \addplot [matrix plot,point meta=explicit] file [] {data/example_powerlaw_cluster.dat};
            \end{axis}
        \end{tikzpicture}
        \captionsetup{justification=centering}
        \caption{Powerlaw Cluster\\$m=8$, $p=0.8$}
    \end{subfigure}\hfill%
    \begin{subfigure}[t]{0.15\textwidth}
        \centering
        \begin{tikzpicture}%
            \begin{axis}[
                colormap name=bluewhitered, 
                point meta min=-2,
                point meta max=+2,
                enlargelimits=false,
                axis on top,
                width=\linewidth,
                height=\linewidth,
                ticks=none,
            ]
                \addplot [matrix plot,point meta=explicit] file [] {data/example_uniform_random_intersection.dat};
            \end{axis}
        \end{tikzpicture}
        \captionsetup{justification=centering}
        \caption{Intersection\\$m=8$, $p=0.4$}
    \end{subfigure}\hfill%
    \begin{subfigure}[t]{0.15\textwidth}
        \centering
        \begin{tikzpicture}%
            \begin{axis}[
                colormap name=bluewhitered, 
                point meta min=-2,
                point meta max=+2,
                enlargelimits=false,
                axis on top,
                width=\linewidth,
                height=\linewidth,
                ticks=none,
            ]
                \addplot [matrix plot,point meta=explicit] file [] {data/example_newman_watts.dat};
            \end{axis}
        \end{tikzpicture}
        \captionsetup{justification=centering}
        \caption{Newman Watts Strogatz\\$k=8$, $p=0.4$}
    \end{subfigure}
    
    \begin{subfigure}[t]{0.15\linewidth}
        \centering
        \includegraphics[width=\linewidth]{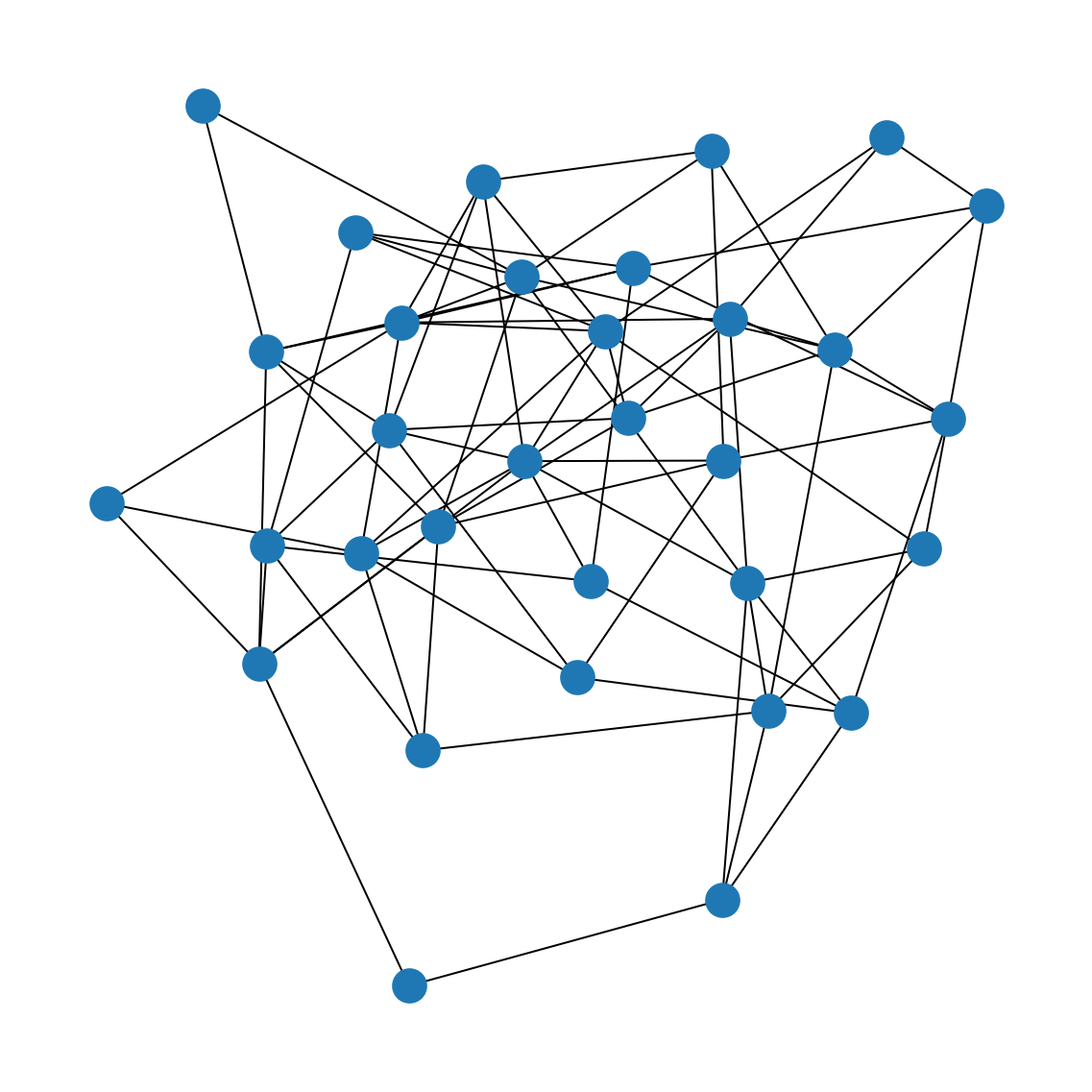}
        \captionsetup{justification=centering}
        \caption{Erdős-Rényi}
    \end{subfigure}\hfill%
    \begin{subfigure}[t]{0.15\linewidth}
        \centering
        \includegraphics[width=\linewidth]{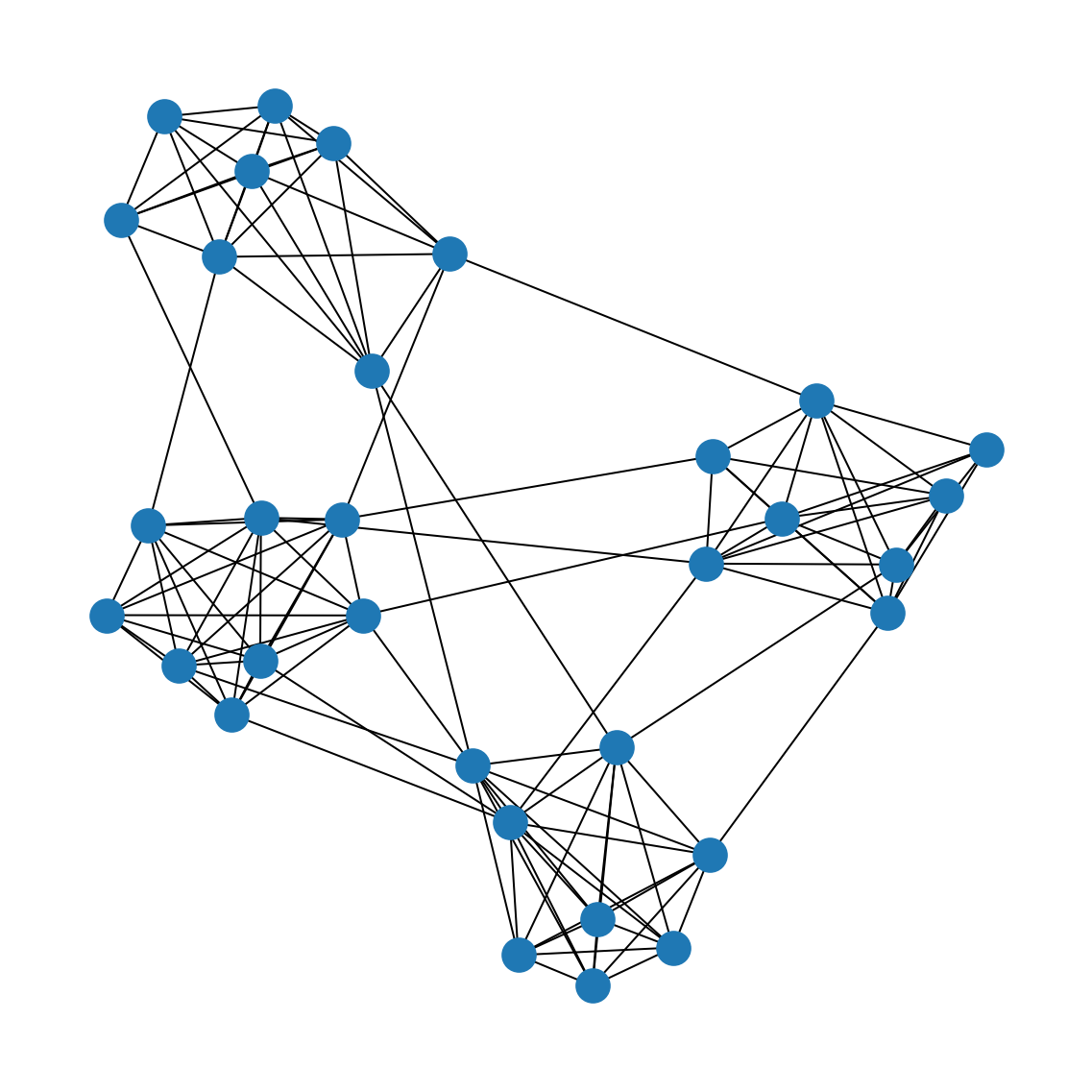}
        \captionsetup{justification=centering}
        \caption{Partition}
    \end{subfigure}\hfill%
    \begin{subfigure}[t]{0.15\linewidth}
        \centering
        \includegraphics[width=\linewidth]{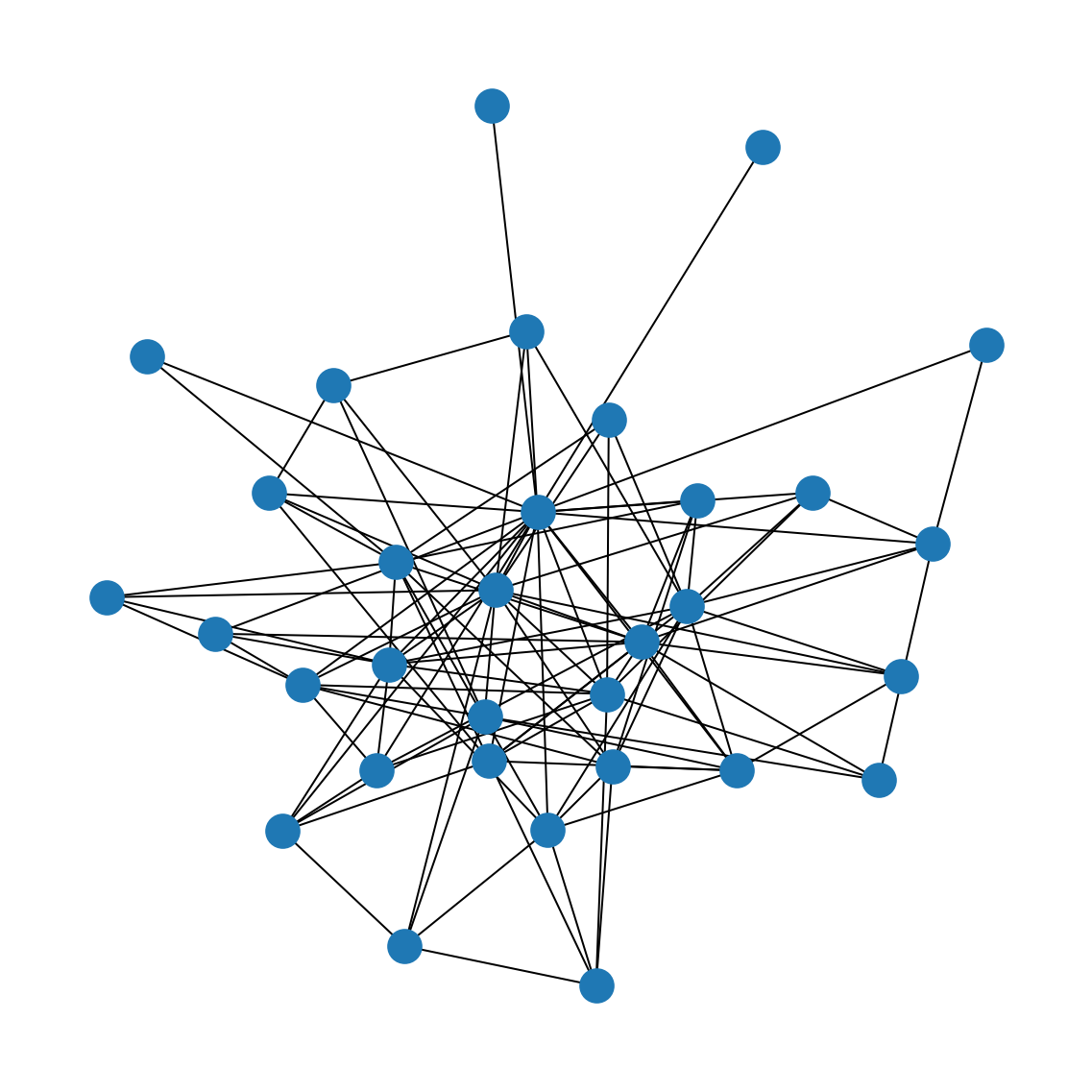}
        \captionsetup{justification=centering}
        \caption{Dual Barabasi Albert}
    \end{subfigure}\hfill%
    \begin{subfigure}[t]{0.15\linewidth}
        \centering
        \includegraphics[width=\linewidth]{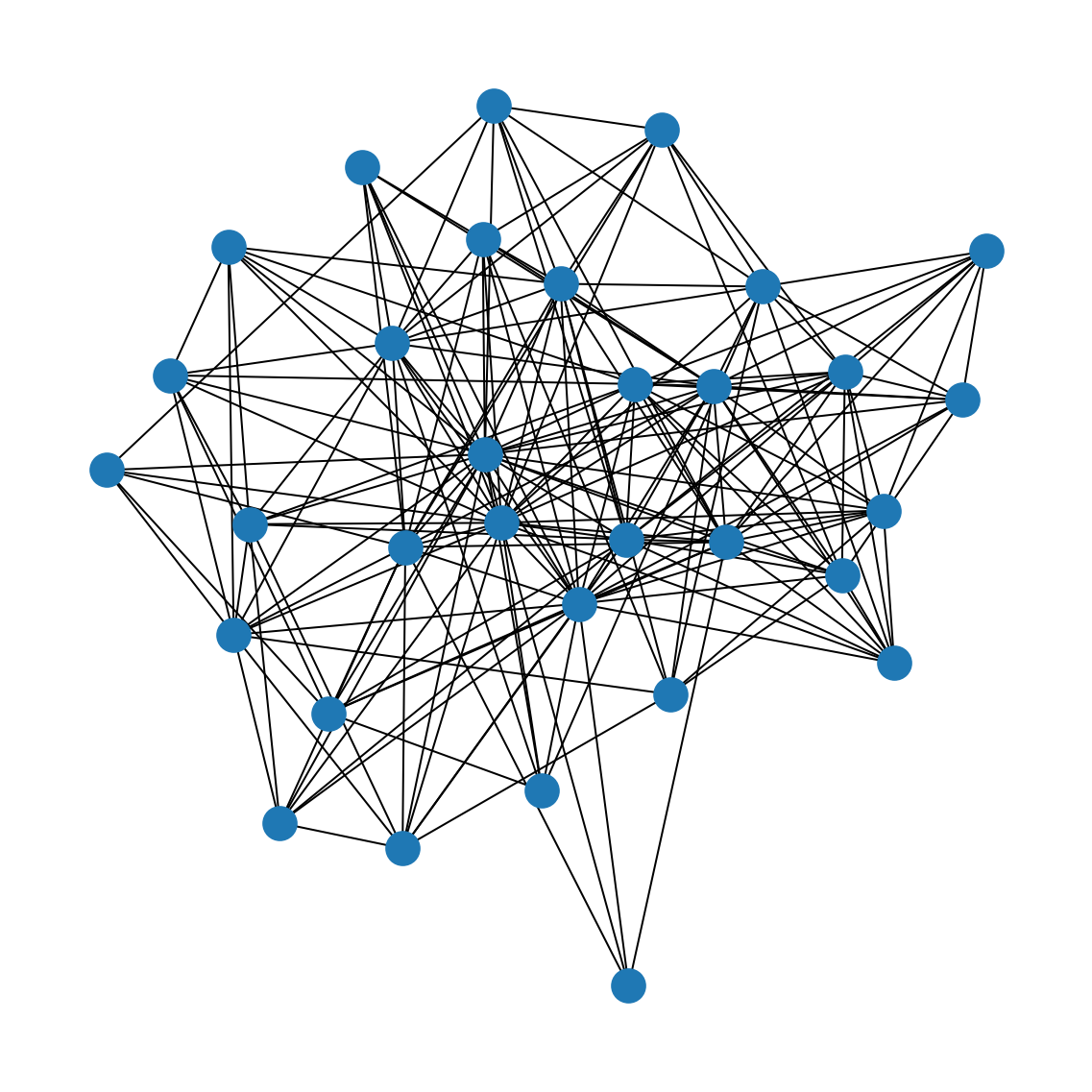}
        \captionsetup{justification=centering}
        \caption{Powerlaw Cluster}
    \end{subfigure}\hfill%
    \begin{subfigure}[t]{0.15\linewidth}
        \centering
        \includegraphics[width=\linewidth]{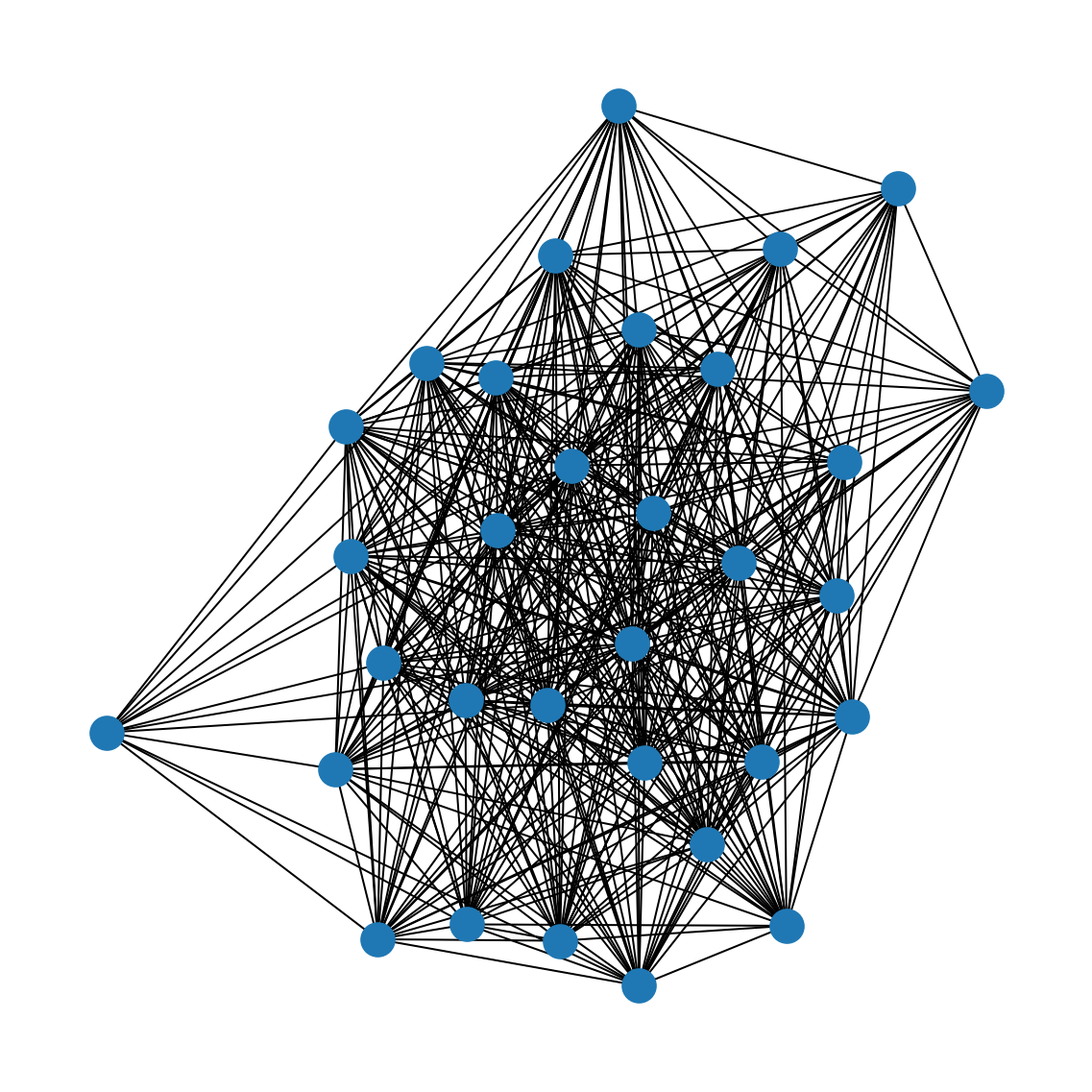}
        \captionsetup{justification=centering}
        \caption{Intersection}
    \end{subfigure}\hfill%
    \begin{subfigure}[t]{0.15\linewidth}
        \centering
        \includegraphics[width=\linewidth]{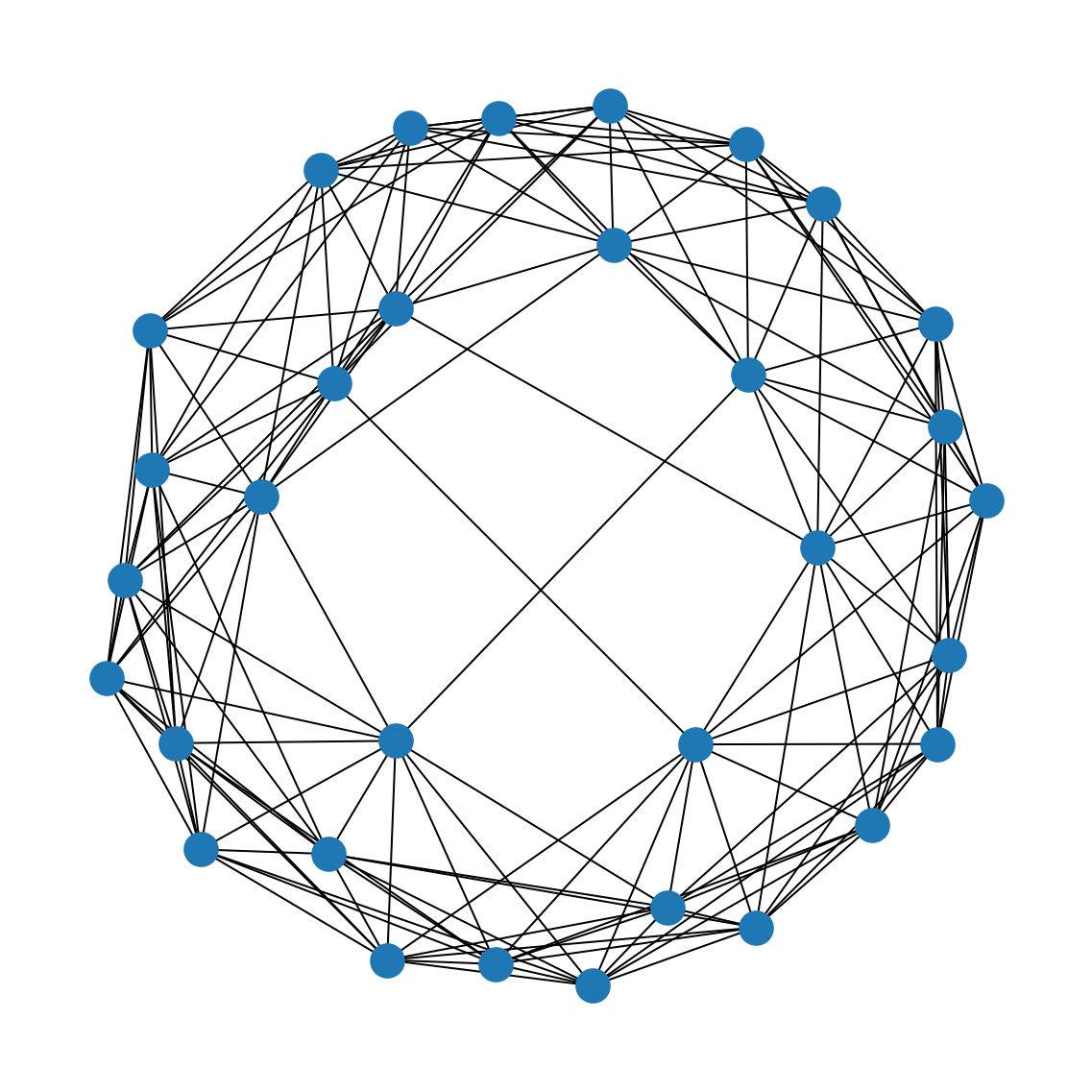}
        \captionsetup{justification=centering}
        \caption{Newman Watts Strogatz}
    \end{subfigure}

    \caption{Example adjacency matrices and vertex and edge plots of different graph classes.}
    \label{fig:example_graphs}
\end{figure*}

\begin{figure*}[t]
    \centering
    \begin{subfigure}[t]{0.49\textwidth}
        \centering
        \caption{Wine}
        \includegraphics[width=0.95\textwidth]{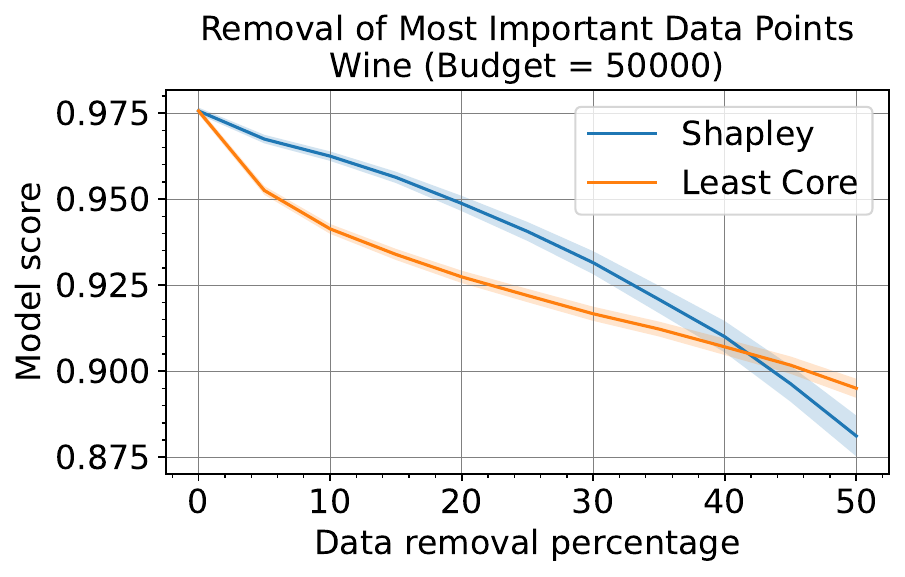}
        \label{fig:xai_data_valuation_wine}
    \end{subfigure}
    \begin{subfigure}[t]{0.49\textwidth}
        \centering
        \caption{Diabetes}
        \includegraphics[width=0.95\textwidth]{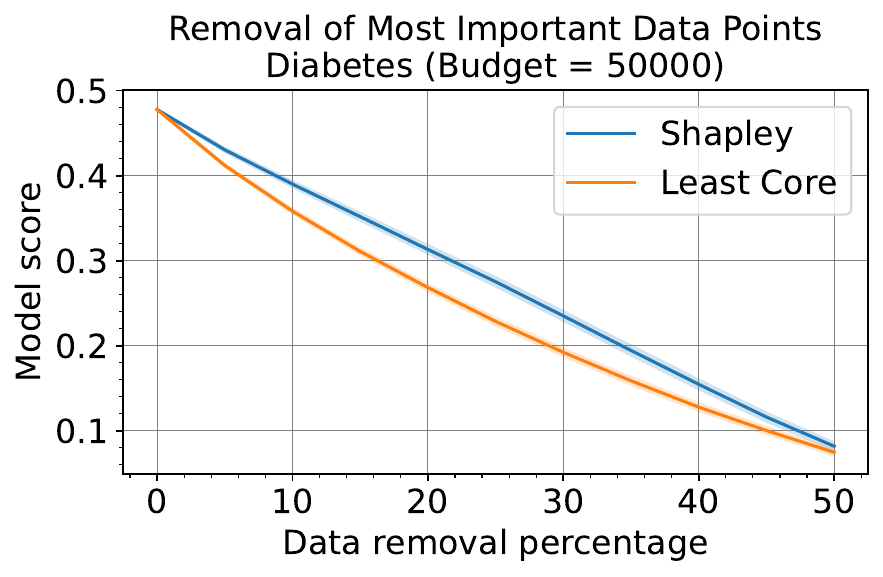}
        \label{fig:xai_data_valuation_diabetes}
    \end{subfigure}
    \caption{Data Valuation Experiments from Section~\ref{sec:empirical_xai:data_valuation} on the Wine and Diabetes datasets.}
    \label{fig:xai_models_appendix}
\end{figure*}

\end{document}